\documentclass[10pt,journal,compsoc]{IEEEtran}
\ifCLASSOPTIONcompsoc
  \usepackage[nocompress]{cite}
\else
  \usepackage{cite}
\fi
\ifCLASSINFOpdf
  \usepackage[pdftex]{graphicx}
  \graphicspath{{./Figures/}}
  \DeclareGraphicsExtensions{.pdf,.jpg,.jpeg,.png,}
\else
\fi
\usepackage{graphicx}
\usepackage{amsmath,amssymb,amsthm} 
\usepackage{color}
\usepackage{multirow}
\usepackage[flushleft]{threeparttable} 
\usepackage{boldline} 
\usepackage[table]{xcolor} 
\usepackage{algorithm,algorithmic,bm,color}
\usepackage[colorlinks=true,allcolors=black]{hyperref}

\newcommand{\argmin}{\arg\!\min}

\newcommand{\Cmat}{{\boldsymbol C}}
\newcommand{\Dmat}{{\boldsymbol D}}

\newcommand{\Hmat}[0]{{{\boldsymbol H}}}
\newcommand{\Imat}{{\boldsymbol I}}

\newcommand{\Rmat}[0]{{{\boldsymbol R}}}

\newcommand{\Tmat}[0]{{{\boldsymbol T}}}

\newcommand{\Xmat}{{\boldsymbol X}}
\newcommand{\Ymat}[0]{{{\boldsymbol Y}}}
\newcommand{\Zmat}{{\boldsymbol Z}}

\newcommand{\uv}[0]{{\boldsymbol{u}}}
\newcommand{\vv}{\boldsymbol{v}}

\newcommand{\xv}{\boldsymbol{x}}
\newcommand{\yv}{\boldsymbol{y}}

\newcommand{\zv}{\boldsymbol{z}}

\newcommand{\thetav}{\boldsymbol{\theta}}

\newcommand{\ts}{^{\mathsf{T}}}
\newcommand{\inv}{^{-1}}
\newcommand{\ie}{{\em i.e.}}
\newcommand{\eg}{{\em e.g.}}

\newtheorem{definition}{Definition}
\newtheorem{theorem}{Theorem}
\newtheorem{corollary}{Corollary}
\newtheorem{assumption}{Assumption}
\newtheorem{lemma}{Lemma}

\begin{document}
%
\title{Plug-and-Play Algorithms for Video\\ Snapshot Compressive Imaging}
%
%
%
%

\author{
        Xin~Yuan,~\IEEEmembership{Senior~Member,~IEEE,}
        Yang~Liu,
        Jinli~Suo,
       Fr\'edo~Durand
        and~Qionghai~Dai,~\IEEEmembership{Senior~Member,~IEEE}
\IEEEcompsocitemizethanks{
\IEEEcompsocthanksitem X.~Yuan is with Bell Labs, Murray Hill, New Jersey, 07974 USA.\protect\\
E-mail: xyuan@bell-labs.com.
\IEEEcompsocthanksitem Y.~Liu and F. Durand are with the Computer
Science and Artificial Intelligence Laboratory, Massachusetts Institute of
Technology, Cambridge, MA 02139, USA.\protect\\
E-mails: \{yliu12,~fredo\}@mit.edu.
\IEEEcompsocthanksitem J.~Suo and Q.~Dai are with the Department of Automation, Tsinghua University, Beijing 100084, China.\protect\\
E-mails: \{jlsuo, qhdai\}@tsinghua.edu.cn.
\IEEEcompsocthanksitem  Our code is available at \url{https://github.com/liuyang12/PnP-SCI_python}.
}
\thanks{Manuscript updated \today.}
}

\IEEEtitleabstractindextext{%
\begin{abstract}
We consider the reconstruction problem of video snapshot compressive imaging (SCI), which captures high-speed videos using a low-speed 2D sensor (detector). The underlying principle of SCI is to modulate sequential high-speed frames with different masks and then these encoded frames are integrated into a snapshot on the sensor and thus the sensor can be of low-speed.
On one hand, video SCI enjoys the advantages of low-bandwidth, low-power and low-cost. On the other hand, applying SCI to large-scale problems (HD or UHD videos) in our daily life is still challenging and one of the bottlenecks lies in the reconstruction algorithm. Exiting algorithms are either too slow (iterative optimization algorithms) or not flexible to the encoding process (deep learning based end-to-end networks).  
In this paper, we develop fast and flexible algorithms for SCI based on the plug-and-play (PnP) framework. In addition to the PnP-ADMM method, 
we further propose the PnP-GAP (generalized alternating projection) algorithm with a lower computational workload. 
We first employ the {\em image} deep denoising priors 
to show that PnP can recover 
a UHD color video with 30 frames from a snapshot measurement.
Since videos have strong temporal correlation, by employing the {\em video} deep denoising priors, we achieve a significant improvement in the results. Furthermore, we extend the proposed PnP algorithms to the color SCI system using mosaic sensors, where each pixel only captures the red, green or blue channels. A joint reconstruction and demosaicing paradigm is developed for flexible and high quality reconstruction of color video SCI systems. 
Extensive results on both simulation and real datasets verify the superiority of our proposed algorithm.
\end{abstract}

}

\maketitle

\section{Introduction\label{Sec:Intro}}
Conventional imaging systems are developed to {\em capture} more data, such as high-resolutions and large field-of-view.
However, to save these captured data, image/video compression methods are immediately applied due to the limited memory and bandwidth. 
This ``capturing images first and processing afterwards" cannot meet the unprecedented demand in recent explosive growth of artificial intelligence and robotics.
To address these challenges, 
computational imaging~\cite{Altmanneaat2298,Mait18CI} constructively combines optics, electronics and algorithms for optimized performance~\cite{BradyNature12,Brady18Optica,Ouyang2018DeepLM} or to provide new abilities~\cite{Brady15AOP,Tsai15OL} to imaging systems. 
Different from conventional imaging, these computational imaging systems usually capture the data in an indirect manner, mostly compressed or coded. 

\begin{figure*}[!htbp]
	\begin{center}
		\includegraphics[width=1\linewidth]{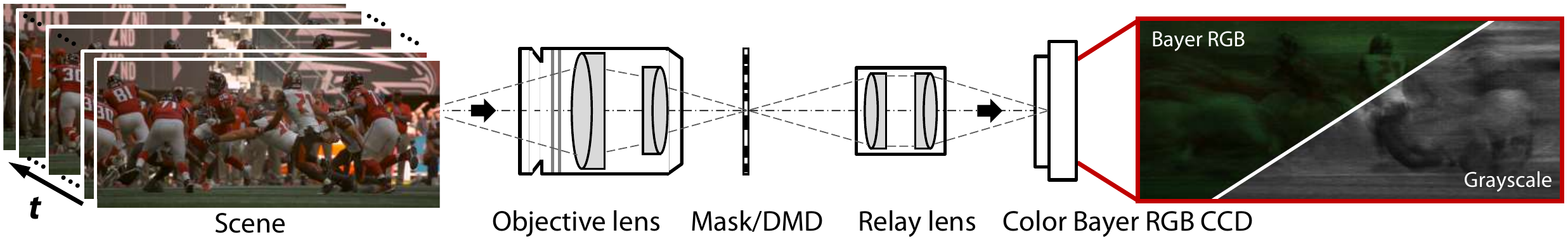}
	\end{center}
	\vspace{-3mm}
	\caption{Schematic of a color video SCI system and its snapshot measurement (showing in Bayer RGB mode). A ``RGGB'' Bayer pattern is shown.}
	\label{fig:video_color_sci}
\end{figure*}

This paper considers one important branch of computational imaging with promising applications, the snapshot compressive imaging (SCI)~\cite{Patrick13OE,Wagadarikar08CASSI}, which utilizes a two-dimensional (2D) camera to capture the 3D  video or spectral data in a snapshot. 
Such imaging systems adopt \emph{compressed sampling} on a set of consecutive images--video frames (\ie, CACTI~\cite{Patrick13OE,Yuan14CVPR}) or spectral channels (\ie, CASSI~\cite{Wagadarikar09CASSI})--in accordance with an encoding procedure and \emph{integrating} these sampled signals along time  or spectrum to obtain the final compressed measurements. With this technique, SCI systems can capture the high-speed motion~\cite{Hitomi11ICCV,Reddy11CVPR,Yuan16BOE,Deng19_sin} and high-resolution spectral information~\cite{Gehm07,Miao19ICCV,Yuan15JSTSP} but with low memory, low bandwidth, low power and potentially low cost. In this work, we focus on video SCI. 

There are two critical challenges in SCI and other computational imaging systems. The first one is the hardware imaging system to capture the {compressed measurements} and the second one is the reconstruction algorithm to retrieve the desired signal. From the encoder-decoder perspective, we call the imaging system ``hardware encoder" 
and the reconstruction algorithm ``software decoder". 

For the first challenge in video SCI, different hardware encoders have been built and the underlying principle is to modulate the high-speed scene with a higher frequency than the sampling speed of the camera (Fig.~\ref{fig:video_color_sci}). Various coding strategies have been proposed, such as using a spatial light modulator (SLM), including a digital micromirror device (DMD)~\cite{Hitomi11ICCV,Qiao2020_APLP,Reddy11CVPR,Sun17OE} or a dynamic mask~\cite{Patrick13OE,Yuan14CVPR}. 
The patterns of DMD will change tens of times during one exposure time of the camera to impose the compression; a physical mask is moving within one exposure time so that different variants of the mask are imposed on the high-speed scenes to achieve the high-speed modulation.

Regarding the second challenge of software decoder, various algorithms have been employed and developed for SCI reconstruction.
In addition to the widely used TwIST~\cite{Bioucas-Dias2007TwIST}, Gaussian Mixture Model (GMM) in~\cite{Yang14GMMonline,Yang14GMM} assumes the pixels within a spatial-temporal patch are drawn  from a GMM. GAP-TV~\cite{Yuan16ICIP_GAP} adopts the idea of total variation (TV)  minimization under the generalized alternating projection (GAP)~\cite{Liao14GAP} framework. 
Recently, DeSCI proposed in~\cite{Liu18TPAMI}  has led to state-of-the-art results.
However, the slow speed of DeSCI precludes its real applications, especially to the HD ($1280\times720$), FHD ($1920\times1080$) or UHD ($3840\times1644$ and $3840\times2160$ in Fig.~\ref{fig:comp_largescale}) videos, which are now commonly used in our daily life.
Recall that DeSCI needs more than one hour to reconstruct a $256\times256\times8$ video from a snapshot measurement.
GAP-TV, by contrast, as a fast algorithm, cannot provide decent reconstructions to be used in real applications (in general, this needs the PSNR $>$30dB).
An alternative solution is to train an end-to-end network to reconstruct the videos~\cite{Ma19ICCV,Qiao2020_APLP,Li2020ICCP,Cheng20ECCV_BIRNAT} for the SCI system. On one hand, this approach can finish the task within seconds (after training) and by the appropriate usage of multiple GPUs, an end-to-end sampling and reconstruction system can be built. On the other hand, this method loses the {\em robustness} of the network since whenever the sensing matrix (encoding process) changes, a new network has to be re-trained. Moreover, it cannot be used in adaptive video sensing~\cite{Yuan13ICIP}.

Therefore, it is desirable to devise an {\em efficient} and {\em flexible} algorithm for video SCI reconstruction, especially for large-scale problems. This will pave the way of applying SCI in our daily life. 
Towards this end, this paper develops plug-and-play (PnP) algorithms for SCI.

\subsection{Related Work \label{Sec:Related}}
From the hardware side, in addition to capture high-speed videos, various other SCI systems have been developed to capture  3D multi/hyper-spectral images~\cite{Cao16SPM,Wang18PAMI,Yuan15JSTSP,Meng2020_OL_SHEM,Meng20ECCV_TSAnet}, 4D spectral-temporal~\cite{Tsai15OL}, spatial-temporal~\cite{Qiao2020_CACTI}, depth~\cite{Llull15Optica,Yuan16AO} and polarization~\cite{Tsai15OE} images, etc. These systems share the similar principle of modulating the high-dimensional signals using high-frequency patterns. 

From the algorithm side, 
early systems usually employed the algorithms for inverse problems of other applications such as compressive sensing~\cite{Candes06ITT,Donoho06ITT}.
In general, the SCI reconstruction is an ill-posed problem and diverse priors and regularization methods have been used.
Among these priors, the TV~\cite{Rudin92_TV}  and sparsity~\cite{Patrick13OE} are widely used. 
Representative algorithms include TwIST~\cite{Bioucas-Dias2007TwIST} and GAP-TV~\cite{Yuan16ICIP_GAP}. 
Recently developed algorithms specifically for SCI include GMM~\cite{Yang14GMMonline,Yang14GMM} and DeSCI~\cite{Liu18TPAMI}, where GMM methods use mixture of Gaussian distributions to model video patches and DeSCI applies weighted nuclear norm minimization~\cite{Gu14CVPR} on video patches into the alternating direction method of multipliers (ADMM)~\cite{Boyd11ADMM} framework.  

As mentioned before, one main bottleneck of these optimization algorithms is the slow running speed. 
Inspired by recent advances of deep learning on image restoration~\cite{zhang2017beyond}, researchers have started using deep learning in computational imaging~\cite{Iliadis18DSPvideoCS,Jin17TIP,Kulkarni2016CVPR,LearningInvert2017,George17lensless,Yuan18OE}. Some networks have been proposed for SCI reconstruction~\cite{Ma19ICCV,Miao19ICCV,Qiao2020_APLP,Yoshida18ECCV,Li2020ICCP}. 
After training, these algorithms can provide results instantaneously and thus they can lead to end-to-end systems~\cite{Meng20ECCV_TSAnet,Qiao2020_CACTI} for SCI.
However, these end-to-end deep learning methods rely heavily on the training data and further are not flexible.
Specifically, when one network is trained for a specific SCI system, it cannot be used in other SCI systems provided different modulation patterns or different compression rates.

In summary, optimization methods are slow but deep learning algorithms are not flexible. To cope with these issues, most recently, researchers start to integrate the advantages of both by applying the deep denoisers into the PnP framework~\cite{Venkatakrishnan_13PnP,Sreehari16PnP,Chan2017PlugandPlayAF,Ryu2019PlugandPlayMP}.
Though PnP can date back to 2013~\cite{Venkatakrishnan_13PnP}, it is getting powerful in real inverse problems because the usage of advanced deep denoising networks~\cite{Zhang17SPM_deepdenoise,Zhang18TIP_FFDNet}. Recently, great successes have been achieved by PnP in other applications. 

Bearing these concerns in mind, in this work, we integrate various denoisers into PnP framework for SCI reconstruction.
Our PnP algorithms can not only provide excellent results but also are robust to different coding process and thus can be used in adaptive sensing and large-scale problems~\cite{Yuan20CVPR}.

\subsection{Contributions of This Work}
Generally speaking, reconstruction of SCI aims to solve the trilemma, \ie, speed, accuracy and flexibility. To address this challenge, our preliminary work~\cite{Yuan20CVPR} applied {\em frame-wise image denoiser} into the PnP framework to achieve excellent results in video SCI. 
Specially, we made the following contributions in~\cite{Yuan20CVPR}.  
\begin{itemize}
	\item[1)] Inspired by the plug-and-play ADMM~\cite{Chan2017PlugandPlayAF} framework, we extend it to SCI and show that PnP-ADMM converges to a fixed point by considering the hardware constraints and the special structure of the sensing matrix in SCI~\cite{Jalali19TIT_SCI}.
	\item[2)] We propose an efficient PnP-GAP algorithm by using various {denoisers} (Fig.~\ref{fig:demo}) into the generalized alternating projection~\cite{Liao14GAP,Yuan16ICIP_GAP} framework, which has a lower computational workload than PnP-ADMM. We prove that, under proper assumptions, the solution of PnP-GAP also converges to a fixed point. 
	\item[3)] 
	By employing the deep image denoiser FFDNet~\cite{Zhang18TIP_FFDNet} into PnP-GAP, 
	we show that a FHD color video (1920$\times$1080$\times$3$\times$30 with 3 denoting the RGB channels and 30 the frame number) can be recovered from a snapshot measurement (Fig.~\ref{fig:comp_largescale}) efficiently {with PSNR close to 30dB}.
	Compared with an end-to-end network~\cite{Qiao2020_APLP}, dramatic resources have been saved since no re-training is required. 
	This further makes the UHD compression using SCI to be feasible (a {3840$\times$1644$\times$3$\times$40 video is reconstructed with PSNR above 30dB in Fig.~\ref{fig:comp_largescale}}). 
	\item[4)] We apply our developed PnP algorithms to extensive simulation and real datasets (captured by real SCI cameras) to verify the efficiency and robustness of our proposed algorithms. We show that the proposed algorithm can obtain results on-par to DeSCI but with a significant reduction of computational time.  
\end{itemize}

Since videos are image sequences and they are highly correlated, it is expected that a {\em video denoiser} will boost up the results of video SCI reconstruction. 
Moreover, the color-video SCI has not been fully exploited by color image/video denoising algorithms.
In particular, we make additional contributions in this paper.
\begin{itemize}
    \item[5)] In addition to the PnP-FFDNet~\cite{Yuan20CVPR}, which integrates the image denoiser, FFDNet~\cite{Zhang18TIP_FFDNet}, into PnP-GAP, we further integrate the most recent video denoiser, FastDVDnet~\cite{Tassano_2020_CVPR}, into PnP-GAP to achieve better results than those reported in \cite{Yuan20CVPR} (Fig.~\ref{fig:demo}). 
    \item[6)] We propose joint reconstruction and demosaicing for color SCI video reconstruction. In color SCI, since each pixel only capture one of the red (R), green (G) or blue (B) channels, previous methods reconstruct each channel (as grayscale images/videos) separately and then use off-the-shelf demosaicing methods to get the final color video. To overcome the limitations of these steps, we jointly reconstruct and demosaic the color video in one shot and better results are achieved. 
    We also build an RGB mid-scale size dataset as benchmark data for the color video SCI problem.
    \item[7)] We verify the proposed PnP-FastDVDnet in the measurements  captured by the newly built SCI camera in~\cite{Qiao2020_APLP} at different compressive sampling rates from 10 to 50. This clearly demonstrates the feasibility and flexibility of the proposed algorithm on the large-scale data. 
    Furthermore, this verifies that a video SCI system can capture high-speed videos at 2500 frames per second (fps) by using a camera working at 50 fps.
\end{itemize}


\begin{figure}[!htbp]
	\begin{center}
		\includegraphics[width=1\linewidth]{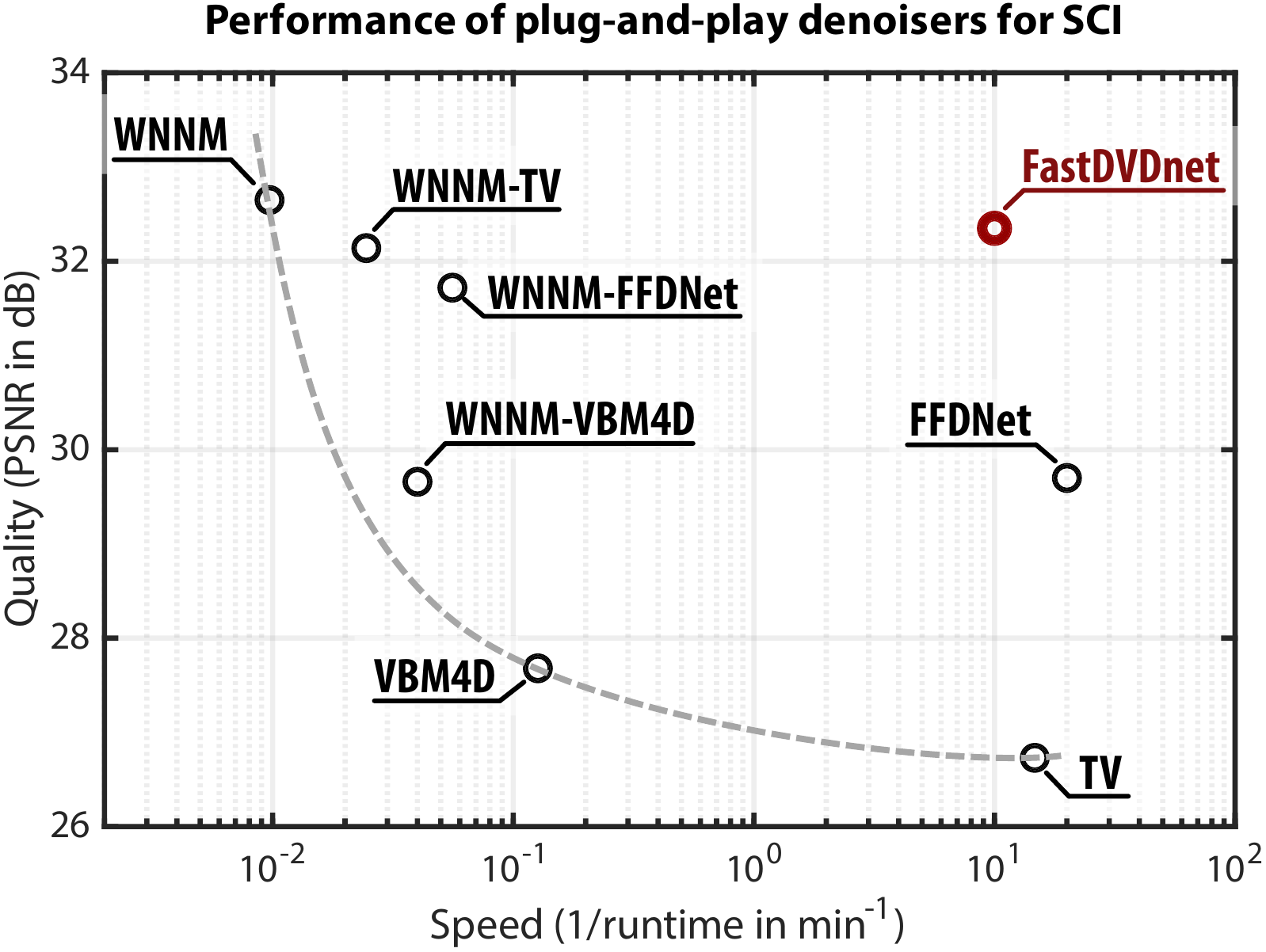}
	\end{center}
	\vspace{-4mm}
	\caption{Trade-off of quality and speed of various plug-and-play denoising algorithms for SCI reconstruction. {Average PSNR of the six grays-scale datasets~\cite{Yuan20CVPR} are shown.}}
	\label{fig:demo}
\end{figure}

\subsection{Organization of This Paper}
The rest of this paper is organized as follows. Sec.~\ref{Sec:SCImodel} introduce the mathematical model of both grayscale and color SCI. Sec.~\ref{Sec:PnP_ADMM} develops the PnP-ADMM under the SCI hardware constraints and shows that PnP-ADMM converges to a fixed point. Sec.~\ref{Sec:PnP_GAP} proposes the PnP-GAP algorithm and proves its convergence\footnote{We observed some error in the proof of the global convergence of PnP-GAP in~\cite{Yuan20CVPR}. Specifically, the lower bound of the second term in Eq. (25) in~\cite{Yuan20CVPR} should be 0. Therefore, the global convergence of PnP-GAP dees not hold anymore. Instead, we provide another convergence proof of PnP-GAP in this paper.}. 
Sec.~\ref{Sec:P3} integrates various denoisers into to the PnP framework for SCI reconstruction and develops the joint demosaicing and reconstruction for color SCI.
Extensive results of both  simulation (grayscale benchmark, mid-scale color and large-scale color) and real data are presented in Sec.~\ref{Sec:results} and Sec.~\ref{Sec:realdata}, respectively.
Sec.~\ref{Sec:Con} concludes the paper.


\begin{figure*}[!htbp]
	\begin{center}
        \includegraphics[width=1\linewidth]{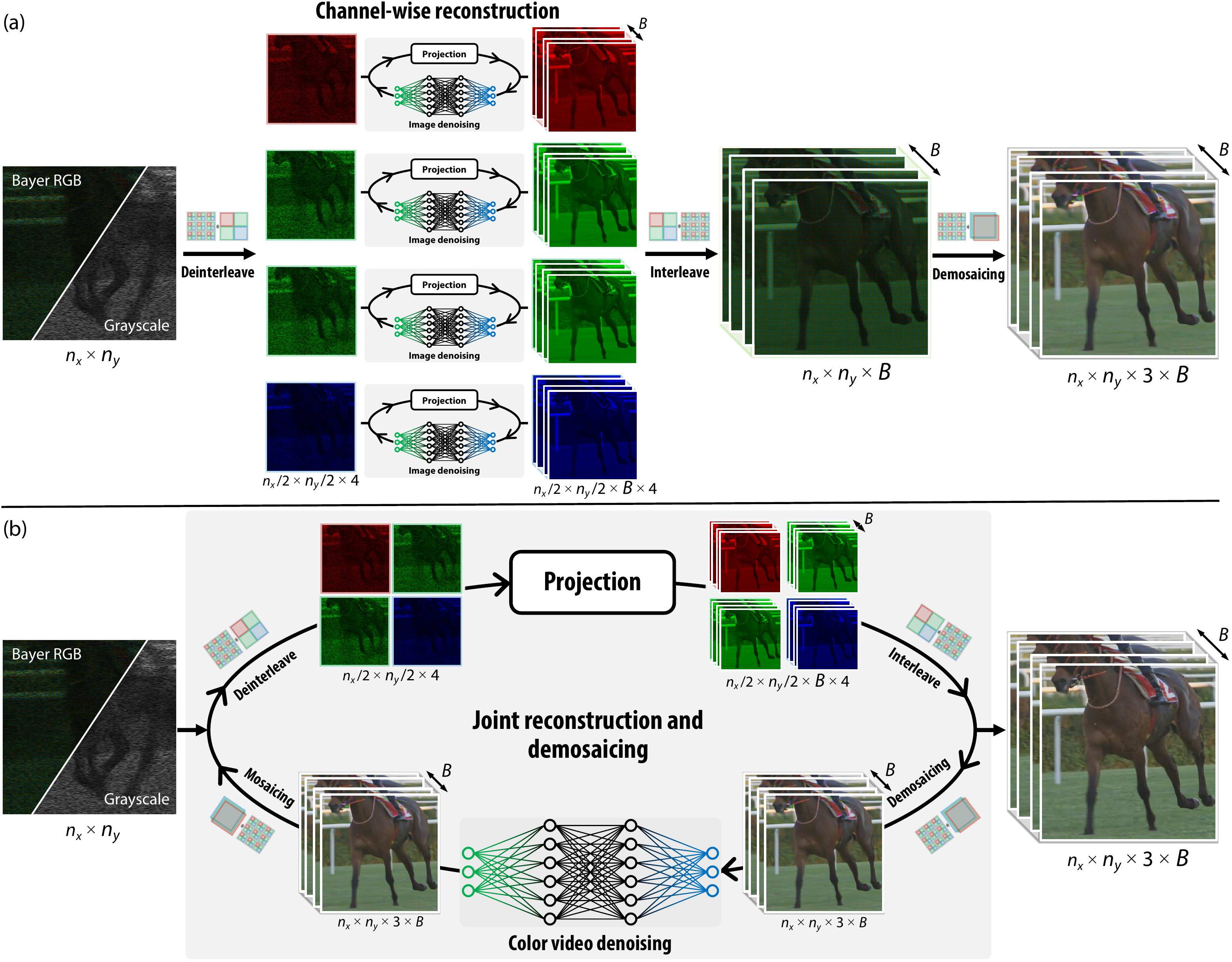}
	\end{center}
	\vspace{-2mm}
	\caption{Reconstruction of color SCI using mosaic sensor measurements. (a) Color SCI reconstruction by independently reconstruct RGGB channels using grayscale {\em image} denoising and then perform demosicing (we proposed this in~\cite{Yuan20CVPR}). The raw measurement (and the mask) is divided into four color channels, R (red), G1 (green), G2 (green) and B (blue) and these channels are reconstructed separately using the PnP-GAP with FFDNet. Then these channels are interleaved and demosaiced to obtain the final color video. (b) Proposed (in this paper) joint reconstruction and demosaicing for color SCI. The raw measurement (and the mask) is sent to the proposed PnP framework using GAP/ADMM with {\em color denoising} by FFDNet or FastDVDnet to output the desired color video directly. Note the demosaicing and {\em color video denoising} are embedded in each iteration.}
	\label{fig:Bayer_sci}
\end{figure*}

\section{Mathematical Model of SCI~\label{Sec:SCImodel}}
As depicted in Fig.~\ref{fig:video_color_sci}, in the video SCI system \eg, CACTI~\cite{Patrick13OE}, consider that a  $B$-frame (grayscale) video $\Xmat \in \mathbb{R}^{n_x \times n_y \times B}$ is modulated  and compressed by $B$ sensing matrices (masks) $\Cmat\in \mathbb{R}^{n_x \times n_y \times B}$, and the measurement frame $\Ymat \in \mathbb{R}^{n_x\times n_y} $ can be expressed as~\cite{Patrick13OE,Yuan14CVPR}
\begin{equation}\label{Eq:System}
  \Ymat = \sum_{b=1}^B \Cmat_b\odot \Xmat_b + \Zmat,
\end{equation}
where $\Zmat \in \mathbb{R}^{n_x \times n_y }$ denotes the noise; $\Cmat_b = \Cmat(:,:,b)$ and $\Xmat_b = \Xmat(:,:,b) \in \mathbb{R}^{n_x \times n_y}$ represent the $b$-th sensing matrix (mask) and the corresponding video frame respectively, and $\odot$ denotes the Hadamard (element-wise) product. 
Mathematically, the measurement in \eqref{Eq:System} can be expressed by 
\begin{equation}\label{Eq:ghf}
\yv = \Hmat \xv + \zv,
\end{equation}
where $\yv = \text{Vec}(\Ymat) \in \mathbb{R}^{n_x n_y}$ and $\zv= \text{Vec}(\Zmat) \in \mathbb{R}^{n_x n_y}$ with $\text{Vec}(\cdot)$ vectorizing the ensued matrix by stacking columns. Correspondingly, the video signal $\xv \in \mathbb{R}^{n_x n_y B}$ is
\begin{equation}
\xv = \text{Vec}(\Xmat) = [\text{Vec}(\Xmat_1)\ts,..., \text{Vec}(\Xmat_B)\ts]\ts.
\end{equation}
Unlike the global transformation based compressive sensing~\cite{Candes05compressed,donoho2006compressed}, the sensing matrix $\Hmat \in \mathbb{R}^{n_x n_y \times n_x n_y B}$ in video SCI is sparse and is constituted by a concatenation of diagonal matrices
\begin{equation}\label{Eq:Hmat_strucutre}
\Hmat = [\Dmat_1,...,\Dmat_B].
\end{equation}
where $\Dmat_b = \text{diag}(\text{Vec}(\Cmat_b)) \in {\mathbb R}^{n \times n}$ with $n = n_x n_y$, for $b =1,\dots B$.
Consequently, the {\em sampling rate} here is equal to  $1/B$. It has been proved recently in~\cite{Jalali18ISIT,Jalali19TIT_SCI} that the reconstruction error of SCI is bounded even when $B>1$.

In the color video case, as shown {in Figs.~\ref{fig:Bayer_sci}, \ref{fig:comp_frames_midscale}, \ref{fig:comp_largescale} and \ref{fig:real_color_hammer}}, the raw data captured by the generally used Bayer pattern sensors have ``RGGB" channels. 
Since the mask is imposed on each pixel, the generated measurement can be treated as a grayscale image as in Fig.~\ref{fig:real_chopperwheel} and when it is shown in color, the demosaicing procedure cannot generate the right color due to mask modulation (Fig.~\ref{fig:Bayer_sci}). 
In previous papers, during reconstruction, we first recover each of these four channel independently and then perform demosaicing in the reconstructed videos (upper part in Fig.~\ref{fig:Bayer_sci}). The final demosaiced RGB video is the desired signal~\cite{Yuan14CVPR,Yuan20CVPR}.
In this case, the raw measurement is decoupled into four components $\{\Ymat^{(r)},\Ymat^{(g_1)},\Ymat^{(g_2)},\Ymat^{(b)}\} \in {\mathbb R}^{\frac{n_x}{2}\times \frac{n_y}{2}}$. Similarly, the corresponding masks and videos are denoted by $\{\Cmat^{(r)},\Cmat^{(g_1)},\Cmat^{(g_2)},\Cmat^{(b)}\} \in {\mathbb R}^{\frac{n_x}{2}\times \frac{n_y}{2}\times B}$, $\{\Xmat^{(r)},\Xmat^{(g_1)},\Xmat^{(g_2)},\Xmat^{(b)}\} \in {\mathbb R}^{\frac{n_x}{2}\times \frac{n_y}{2}\times B}$, respectively. 
The forward model for each channel is now
\begin{eqnarray}
 \Ymat^{(r)} &=&  \sum_{b=1}^B \Cmat^{(r)}_b\odot \Xmat^{(r)}_b + \Zmat^{(r)}, \\
 \Ymat^{(g_1)} &=& \sum_{b=1}^B \Cmat^{(g_1)}_b\odot \Xmat^{(g_1)}_b + \Zmat^{(g_1)}, \\
 \Ymat^{(g_2)} &=& \sum_{b=1}^B \Cmat^{(g_2)}_b\odot \Xmat^{(g_2)}_b + \Zmat^{(g_2)}, \\
 \Ymat^{(b)} &=& \sum_{b=1}^B \Cmat^{(b)}_b\odot \Xmat^{(b)}_b + \Zmat^{(b)}. 
\end{eqnarray}
In this color case, the desired signal is $\Xmat^{(rgb)}\in \mathbb {R}^{n_x\times n_y\times 3\times B}$, where $3$ denotes the R, G and B channels in the color video. The demosaicing is basically an interpolation process from $\Xmat^{(r)}$  to $\tilde{\Xmat}^{(r)} \in  \mathbb {R}^{n_x\times n_y\times B}$, from $\{\Xmat^{(g_1)},\Xmat^{(g_2)}\}$  to $\tilde{\Xmat}^{(g)} \in  \mathbb {R}^{n_x\times n_y\times B}$ and from $\Xmat^{(b)}$  to $\tilde{\Xmat}^{(b)} \in  \mathbb {R}^{n_x\times n_y\times B}$.
Note that the interpolation rate for red and blue channel is from 1 pixel to 4 pixels, whereas for the green channel it is from 2 pixels to 4 pixels.

Utilizing the vectorized formulation, let $\{\tilde{\xv}^{(r)},\tilde{\xv}^{(g)},\tilde{\xv}^{(b)}\}\in {\mathbb R}^{n_xn_yB}$ denote the vectorized representations of $\{\tilde{\Xmat}^{(r)},\tilde{\Xmat}^{(g)},\tilde{\Xmat}^{(b)}\}$, $\{{\yv}^{(r)},{\yv}^{(b)}\}\in {\mathbb R}^{\frac{n_xn_y}{4}}$ denote the vectorized representations of $\Ymat^{(r)},{\Ymat}^{(b)}\}$ and $\yv^{(g)} = \left[\begin{array}{c}\yv^{(g_1)}\\ \yv^{(g_2)}\end{array}\right]\in {\mathbb R}^{\frac{n_xn_y}{2}} $ denote the concatenated vector-representation of $\{\Ymat^{(g_1)} ,\Ymat^{(g_2)} \}$. Similar notations are also used for the noise term. We arrive at
\begin{eqnarray}
 \yv^{(r)}&=& \Hmat^{(r)} \tilde{\xv}^{(r)} + \zv^{(r)},\\
 \yv^{(g)}&=& \Hmat^{(g)} \tilde{\xv}^{(g)} + \zv^{(g)},\\
 \yv^{(b)}&=& \Hmat^{(b)} \tilde{\xv}^{(b)}+ \zv^{(b)},
\end{eqnarray}
where $\{\Hmat^{(r)},\Hmat^{(b)}\} \in {\mathbb R}^{\frac{n_xn_y}{4} \times n_xn_yB}$
and $\Hmat^{(g)}\in {\mathbb R}^{\frac{n_xn_y}{2} \times n_xn_yB}$. The structures of $\{\Hmat^{(r)},\Hmat^{(g)},\Hmat^{(b)}\}$ are similar to \eqref{Eq:Hmat_strucutre} for grayscale video but include the down-sampling (mosaic) process, which decimate pixels in an interleaving way following the mosaic pattern of the sensor.  


Following this, let the captured mosaic compressed measurement and the desired color video be $\yv \in{\mathbb R}^{n_xn_y}$ and $\xv \in{\mathbb R}^{3n_xn_yB}$, respectively. We have
\begin{eqnarray}
 \yv = \left[\begin{array}{l}
 \yv^{(r)} \\
  \yv^{(g)}\\
    \yv^{(b)} \end{array}\right], \quad   \xv = \left[\begin{array}{l}
 \tilde{\xv}^{(r)} \\
  \tilde{\xv}^{(g)}\\
    \tilde{\xv}^{(b)} \end{array}\right],
\end{eqnarray}
and the full froward model of color-video SCI can be modeled by
\begin{eqnarray} \label{eq:forward_color}
 \yv = \underbrace{\left[\begin{array}{ccc}
      \Hmat^{(r)} & {\bf 0} & {\bf 0} \\
      {\bf 0}&  \Hmat^{(g)} &{\bf 0} \\
       {\bf 0}&  {\bf 0}& \Hmat^{(b)}
 \end{array}\right]}_{\Hmat} \xv + \zv.
\end{eqnarray}
This formulation along with the grayscale one is the unique forward model of video SCI and apparently, the color one is more challenging. 

As depicted in the top-part of Fig.~\ref{fig:Bayer_sci}, previous studies usually first reconstruct the four Bayer channels of the video independently and then employ the off-the-shelf demosaicing algorithm to get the desired color videos~\cite{Liu18TPAMI,Yuan20CVPR}.
However, this final performance of the reconstructed video will be limited by both steps (channel-wise reconstruction and demosaicing). 
In this paper, we derive a joint reconstruction and demosaicing framework for color video SCI (lower part in Fig.~\ref{fig:Bayer_sci}) directly based on the model derived in \eqref{eq:forward_color}.  More importantly, it is a unified PnP framework, where different demosaicing and color denoising algorithms can be used. Please refer to the details in Sec.~\ref{Sec:jointcsci}.

\section{Plug-and-Play ADMM for SCI~\label{Sec:PnP_ADMM}}
The inversion problem of SCI can be modeled as
\begin{equation}
  {\hat \xv} = \argmin_{\xv} f(\xv) + \lambda g(\xv), \label{Eq:uncontr}
\end{equation}
where $f(\xv)$ can be seen as the forward imaging model, \ie, $\|\yv-\Hmat\xv\|_2^2$ and $g(\xv)$ is a prior being used. This prior is usually playing the role of a regularizer.
While diverse priors have been used in SCI such as TV, sparsity and low-rank, in this work, we focus on the deep denoising prior, which has shown superiority recently on various image restoration tasks.
Note that, since SCI systems aim to reconstruct high-speed video, a video deep denoising prior is desired~\cite{Tassano_19ICIP_DVDnet,Tassano_2020_CVPR}. On the other hand, since videos are essentially consequent images, recently advanced deep denoising priors for images can also be used~\cite{Zhang17SPM_deepdenoise,Zhang18TIP_FFDNet}.   
It has been shown in our preliminary paper that an efficient image denoising prior can lead to good results for SCI~\cite{Yuan20CVPR}. However, this frame-wise image denoising prior~\cite{Zhang18TIP_FFDNet} limits the performance of video denoising since it ignored the strong temporal correlation in neighbouring frames. In this work, we employ the most recent video denoiser, FastDVDnet~\cite{Tassano_2020_CVPR}, as the denoising prior in our PnP framework and it leads to better results than those reported in~\cite{Yuan20CVPR}.

\subsection{Review the Plug-and-Play ADMM}
Using ADMM~\cite{Boyd11ADMM}, by introducing an auxiliary parameter $\vv$, the unconstrained optimization in Eq.~\eqref{Eq:uncontr} can be converted into
\begin{equation} \label{Eq:ADMM_xv}
({\hat \xv}, {\hat \vv}) = \argmin_{\xv,\vv} f(\xv) + \lambda g(\vv), {\text{ subject to }} \xv = \vv.
\end{equation}  
This minimization can be solved by the following sequence of sub-problems~\cite{Chan2017PlugandPlayAF}
\begin{align}
\xv^{(k+1)} &=  \argmin_{\xv} f(\xv) + \frac{\rho}{2} \|\xv - (\vv^{(k)}-\frac{1}{\rho} \uv^{(k)})\|_2^2,  \label{Eq:solvex}\\
\vv^{(k+1)} &=   \argmin_{\vv} \lambda g(\vv) + \frac{\rho}{2}\|\vv - (\xv^{(k)}+\frac{1}{\rho} \uv^{(k)})\|_2^2, \label{Eq:solvev}\\
\uv^{(k+1)} &= \uv^{(k)} + \rho (\xv^{(k+1)} - \vv^{(k+1)}), \label{Eq:u_k+1}
\end{align}
where the superscript $^{(k)}$ denotes the iteration number.

In SCI and other inversion problems, $f(\xv)$ is usually a quadratic form and there are various solutions to Eq.~\eqref{Eq:solvex}. In PnP-ADMM, the solution of Eq.~\eqref{Eq:solvev} is replaced by an {\em off-the-shelf} denoising algorithm, to yield
\begin{equation}
{  \vv^{(k+1)} = {\cal D}_{\sigma_k} (\xv^{(k)}+\frac{1}{\rho} \uv^{(k)})}.
\end{equation}
where ${\cal D}_{\sigma_k}$ denotes the denoiser being used with $\sigma_k$ being the standard deviation of the assumed additive white  Gaussian noise in the $k$-th iteration.
In~\cite{Chan2017PlugandPlayAF}, the authors proposed to update the $\rho$ in each iteration by $\rho_{k+1} = \gamma_k \rho_k$ with $\gamma_k \ge 1$ and setting $\sigma_k = \sqrt{\lambda/\rho_k}$ for the denoiser. 
This essentially imposed the {\em non-increasing denoiser} in Assumption~\ref{Ass:non_in} defined in Sec.~\ref{Sec:Conv_gap}. 
Chen {\em et al.}~\cite{Chan2017PlugandPlayAF} defined the {\em bounded denoiser} and proved the {\em fixed point} convergence of the PnP-ADMM.

\begin{definition} 
(Bounded Denoiser~\cite{Chan2017PlugandPlayAF}): A bounded denoiser with a parameter $\sigma$ is a function 
${\cal D}_{\sigma}: {\mathbb R}^n \rightarrow {\mathbb R}^n$ such that for any input $\xv\in {\mathbb R}^{n}$, 
	\begin{equation}
  	\frac{1}{n}\|{\cal D}_{\sigma}(\xv) - \xv\|_2^2 \le \sigma^2 C,
	\end{equation}
	for some universal constant $C$ independent of $n$ and $\sigma$. 
	\label{Definition1}
\end{definition}
With this definition (constraint on the denoiser) and the assumption of $f:[0,1]^n \rightarrow {\mathbb R}$ having bounded gradient, which is for any $\xv \in [0,1]^n$, there exists $L < \infty$ such that $\|\nabla f(\xv)\|_2/\sqrt{n} \le L$, the authors of~\cite{Chan2017PlugandPlayAF} have proved that: 
the iterates of the PnP-ADMM demonstrates a fixed-point convergence. That is, 
there exists $(\xv^*, \vv^*, \uv^*)$ such that $\|\xv^{(k)} - \xv^*\|_2 \rightarrow 0$, $\|\vv^{(k)} - \vv^*\|_2 \rightarrow 0$, and $\|\uv^{(k)} - \uv^*\|_2 \rightarrow 0$ as $ k\rightarrow \infty$.
%

\subsection{PnP-ADMM for SCI}
In the following  derivation, we focus on the grayscale case and it is ready to extend to the color SCI cases.
In SCI, with the model stated in Eq.~\eqref{Eq:ghf}, $\xv \in {\mathbb R}^{nB}$, we consider the loss function $f(\xv)$ as
\begin{equation}
f(\xv) = \frac{1}{2}\|\yv - \Hmat \xv\|_2^2.
\end{equation}
Consider all the pixel values are normalized into $[0,1]$. 

\begin{lemma} In SCI, the function $f(\xv) = \frac{1}{2}\|\yv-\Hmat\xv\|_2^2$ has bounded gradients, \ie, $\|\nabla f(\xv)\|_2\leq B \|\xv\|_2$. 
	\label{Lemma:fx_grad}
\end{lemma}
\begin{proof}
	The gradient of $f(\xv)$ in SCI is 
\begin{equation}
\nabla f(\xv) = \Hmat\ts\Hmat\xv-\Hmat\ts\yv,
\end{equation} where $\Hmat$ is a concatenation of diagonal matrices of size $n\times nB$ as shown in Eq.~\eqref{Eq:Hmat_strucutre}. 
\begin{list}{\labelitemi}{\leftmargin=12pt \topsep=0pt \parsep=0pt}
	\item The $\Hmat\ts\yv$ is a non-negative constant since both the measurement $\yv$ and the mask are non-negative in nature.
	\item Now let's focus on $\Hmat\ts\Hmat\xv$. Since
	\begin{align}
	\label{eq_sesci_PTP}
	\Hmat\ts\Hmat&=
	\left[
	\begin{matrix}
	\Dmat_1 \\ 
	\vdots \\ 
	\Dmat_B
	\end{matrix}
	\right] \left[
	\begin{matrix}
	\Dmat_1 \cdots \Dmat_B
	\end{matrix}
	\right]\\
	& = \left[
	\begin{matrix}
	\Dmat_1^2& \Dmat_1\Dmat_2 & \cdots & \Dmat_1\Dmat_B\\
	\Dmat_1 \Dmat_2& \Dmat^2_2 & \cdots & \Dmat_2\Dmat_B\\
	\vdots & \vdots & \ddots & \vdots\\
	\Dmat_1 \Dmat_B& \Dmat_2\Dmat_B & \cdots & \Dmat^2_B
	\end{matrix}
	\right],
	\end{align}
\end{list}
due to this special structure, $\Hmat\ts\Hmat\xv$ is the weighted sum of the $\xv$ and $\|\Hmat\ts\Hmat\xv\|_2\leq B C_{\rm max}\|\xv\|_2$, where $C_{\rm max}$ is the maximum value in the sensing matrix. Usually, the sensing matrix is normalized to $[0,1]$ and this leads to $C_{\rm max}=1$ and therefore $\|\Hmat\ts\Hmat\xv\|_2\leq B \|\xv\|_2$.

Thus, $\nabla f(\xv)$ is bounded.
Furthermore, 
\begin{itemize}
	\item If the mask element $D_{i,j}$ is drawn from a binary distribution with entries \{0,1\} with a property of $p_1 \in (0,1)$ being 1, then
	\begin{eqnarray}
	\|\Hmat\ts\Hmat\xv\|_2\leq p_1 B \|\xv\|_2
	\end{eqnarray}
	with {\em a high probability}; usually, $p_1 = 0.5$ and thus $\|\Hmat\ts\Hmat\xv\|_2\leq 0.5 B \|\xv\|_2$.
	\item If the mask element $D_{i,j}$ is drawn from a Gaussian distribution ${\cal N}(0, \sigma^2)$ as in~\cite{Jalali18ISIT,Jalali19TIT_SCI}, though it is not practical to get negative modulation (values of $D_{i,j}$) in hardware, 
	\begin{eqnarray}
	\|\Hmat\ts\Hmat\xv\|_2\leq  B\sigma^2 \|\xv\|_2\stackrel{\sigma = 1}{=} B\|\xv\|_2,
	\end{eqnarray}
	with {\em a high probability}, where the concentration of measure is used.
\end{itemize}
\end{proof}

%

%
	%
Lemma~\ref{Lemma:fx_grad} along with the bounded denoiser in Definition~\ref{Definition1} gives us the following Corollary.
\begin{corollary}
\label{Coro1}
	Consider the sensing model of SCI in \eqref{Eq:ghf}. Given $\{\Hmat,\yv\}$, ${\xv}$ is solved iteratively via PnP-ADMM with bounded denoiser, then $\xv^{(k)}$ and $\thetav^{(k)}$ will converge to a fixed point.
\end{corollary}
\begin{proof}
The proof follows \cite{Chan2017PlugandPlayAF} and thus omitted here.
\end{proof}

\section{Plug-and-Play GAP for SCI \label{Sec:PnP_GAP}}
In this section, following the generalized alternating projection (GAP) algorithm~\cite{Liao14GAP} and the above conditions on PnP-ADMM, we propose the PnP-GAP for SCI, which has a lower computational workload (thus faster) than PnP-ADMM. 

\subsection{Algorithm}
Different from the ADMM in Eq.~\eqref{Eq:ADMM_xv}, GAP solves SCI by the following problem
\begin{equation} \label{Eq:GAP_xv}
({\hat \xv}, {\hat \vv}) = \argmin_{\xv,\vv} \frac{1}{2}\|\xv - \vv\|_2^2 + \lambda g(\vv), ~{\text{s.t.}}~~ \yv = \Hmat\xv.
\end{equation}
Similarly to ADMM, the minimizer in Eq.~\eqref{Eq:GAP_xv} is solved by a sequence of subproblems and we again let $k$ denotes the iteration number.
\begin{list}{\labelitemi}{\leftmargin=10pt \topsep=0pt \parsep=0pt}
	\item Solving $\xv$: given $\vv$, $\xv^{(k+1)}$ is updated via an Euclidean projection of
	$\vv^{(k)}$ on the linear manifold ${\cal M}: \yv = \Hmat \xv$,
	\begin{equation}
	\xv^{(k+1)} =  \vv^{(k)} + \Hmat\ts (\Hmat \Hmat\ts)\inv (\yv - \Hmat \vv^{(k)}). \label{Eq:x_k+1}
	\end{equation}
	Recall~\eqref{Eq:Hmat_strucutre}, $\{\Dmat_i\}_{i=1}^B$ is a diagonal matrix
\begin{equation}
\Dmat_i = {\rm diag} (D_{i,1}, \dots, D_{i,n}). \nonumber
\end{equation}
Thereby, $\Hmat\Hmat\ts$ is diagonal matrix, \ie,
\begin{equation}
{\Rmat = \Hmat\Hmat\ts = {\rm diag}(R_1, \dots, R_n)},\label{eq:R}
\end{equation}
where $ R_{j} = \sum_{b=1}^{B} D^2_{i,j}, \forall j = 1,\dots,n$.
Eq.~\eqref{Eq:x_k+1} can thus be solved efficiently.
	\item Solving $\vv$: given $\xv$, updating $\vv$ can be seen as a denoising problem and
	\begin{equation}
	{  \vv^{(k+1)} = {\cal D}_{\sigma}(\xv^{(k+1)}).} \label{Eq:Denoise_GAP}
	\end{equation}
	Here, various denoiser can be used with $\sigma = \sqrt{\lambda}$.
\end{list}  

\begin{algorithm}[!htbp]
	\caption{Plug-and-Play GAP}
	\begin{algorithmic}[1]
		\REQUIRE$\Hmat$, $\yv$.
		\STATE Initial $\vv^{(0)}$, $\lambda_0$, $\xi<1$.
		\WHILE{Not Converge}
		\STATE Update $\xv$ by Eq.~\eqref{Eq:x_k+1}. 
		\STATE Update $\vv$ by denoiser  $\vv^{(k+1)} = {\cal D}_{\sigma_k}(\xv^{(k+1)})$.
		\IF {$\Delta_{k+1}\ge \eta \Delta_k$}
		\STATE {$\lambda_{k+1} = \xi \lambda_k$}
		\ELSE 
		\STATE {$\lambda_{k+1} =  \lambda_k$}
		\ENDIF
		\ENDWHILE
	\end{algorithmic}
	\label{algo:PP_GAP}
\end{algorithm}

We can see that in each iteration, the only parameter to be tuned is $\lambda$ and we thus set $\lambda_{k+1} = \xi_k \lambda_k$ with $\xi_k\le 1$.
Inspired by the PnP-ADMM, we update $\lambda$ by the following two rules:
\begin{list}{\labelitemi}{\leftmargin=12pt \topsep=0pt \parsep=0pt}
	\item [a)] Monotone update by setting
	$\lambda_{k+1} = \xi \lambda_k$, with $\xi<1$. 
	\item [b)] Adaptive update
  by considering the relative residue:
	\begin{eqnarray}
	{\textstyle  \Delta_{k+1} = \frac{1}{\sqrt{nB}}\left(\|\xv^{(k+1)} - \xv^{(k)}\|_2 + \|\vv^{(k+1)} - \vv^{(k)}\|_2\right)}.\nonumber \label{eq:Delta}
	\end{eqnarray}
	For any $\eta \in [0,1)$ and let $\xi<1$ be a constant, $\lambda_k$ is conditionally updated according  to the following settings:
	\begin{list}{\labelitemi}{\leftmargin=14pt \topsep=0pt \parsep=0pt}
		\item [i)] If $\Delta_{k+1}\ge \eta \Delta_k$, then $\lambda_{k+1} = \xi \lambda_k$.
		\item [ii)] If $\Delta_{k+1}< \eta \Delta_k$, then $\lambda_{k+1} =  \lambda_k$.
	\end{list}
\end{list}
With this adaptive updating of $\lambda_k$, the full PnP-GAP algorithm for SCI is exhibited in Algorithm~\ref{algo:PP_GAP}.

\subsection{Convergence \label{Sec:Conv_gap}}
\begin{assumption}\label{Ass:non_in}
	(Non-increasing denoiser) The denoiser in each iteration of PnP-GAP ${\cal D}_{\sigma_{k}}: {\mathbb R}^{nB} \rightarrow {\mathbb R}^{nB}$ performs denoising in a non-increasing order, \ie, $\sigma_{k+1}\le \sigma_k$. Further, when $k\rightarrow+\infty$, $\sigma_k \rightarrow 0$.
\end{assumption} 
This assumption makes sense since as the algorithm proceeds we expect the algorithm's estimate of the underlying signal to become more accurate, which means that the denoiser needs to deal with a less noisy signal.  
This is also guaranteed by the $\lambda$ setting in Algorithm~\ref{algo:PP_GAP} and imposed by $\rho$ setting in the PnP-ADMM~\cite{Chan2017PlugandPlayAF}. 
With this assumption, we have the following convergence result of PnP-GAP.
\begin{theorem} \label{The:GAP_SCI_bound}
Consider the sensing model of SCI. Given $\{\Hmat,\yv\}$, ${\xv}$ is solved by PnP-GAP with bounded denoiser in a non-increasing order, then $\xv^{(k)}$ converges.
\end{theorem}
\begin{proof}
	From \eqref{Eq:x_k+1}, $\xv^{(k+1)} =  \vv^{(k)} + \Hmat\ts (\Hmat \Hmat\ts)\inv (\yv - \Hmat \vv^{(k)})$, we have
	\begin{equation}
	\xv^{(k+1)}-\xv^{(k)} =  \vv^{(k)}-\xv^{(k)} + \Hmat\ts \Rmat\inv (\yv - \Hmat \vv^{(k)}). 
	\end{equation}
	Following this, 
	\begin{align}
	&\|\xv^{(k+1)} - \xv^{(k)}\|_2^2 \nonumber\\
	=&\|\vv^{(k)} + \Hmat\ts \Rmat\inv (\yv - \Hmat \vv^{(k)}) - \xv^{(k)} \|^2_2 \\
	=& \|\vv^{(k)} + \Hmat\ts \Rmat\inv (\Hmat\xv^{(k)} - \Hmat \vv^{(k)}) - \xv^{(k)} \|^2_2 \nonumber\\
	=& \|(\Imat - \Hmat\ts \Rmat\inv \Hmat) (\vv^{(k)} - \xv^{(k)})\|^2_2  \label{eq:mid_step}\\
	=& \|\vv^{(k)} - \xv^{(k)}\|_2^2 - \|\Rmat^{-\frac{1}{2}}\Hmat (\vv^{(k)} - \xv^{(k)})\|_2^2 \label{Eq:xk_vkminus1}\\
	\le & \|\vv^{(k)} - \xv^{(k)}\|_2^2 \\
	=&  \| {\cal D}_{\sigma_{k}} (\xv^{(k)}) - \xv^{(k)}\|_2^2  \\
	\le&  \sigma_k^2 nBC \label{Eq:convg_C},
	\end{align} 
	where $\Rmat = \Hmat\Hmat\ts$ as defined in \eqref{eq:R} and the following shows the derivation from \eqref{eq:mid_step} to \eqref{Eq:xk_vkminus1}
	\begin{align}
        &\|(\Imat - \Hmat\ts \Rmat\inv \Hmat) (\vv^{(k)} - \xv^{(k)})\|^2_2 \nonumber\\
        &= (\vv^{(k)} - \xv^{(k)})\ts [\Imat - \Hmat\ts (\Hmat \Hmat\ts)\inv \Hmat]\ts \nonumber\\
        &\quad\, \cdot [\Imat - \Hmat\ts (\Hmat \Hmat\ts)\inv \Hmat](\vv^{(k)} - \xv^{(k)}) \nonumber\\
        &= (\vv^{(k)} - \xv^{(k)})\ts [\Imat - 2\Hmat\ts \Rmat\inv \Hmat +\Hmat\ts \Rmat\inv\Rmat \Rmat\inv \Hmat ]\nonumber\\
        &\quad\, \cdot (\vv^{(k)} - \xv^{(k)}) \nonumber\\
        & =  (\vv^{(k)} - \xv^{(k)})\ts (\Imat - \Hmat\ts \Rmat\inv \Hmat )(\vv^{(k)} - \xv^{(k)}) \nonumber\\
        &= \|\vv^{(k)} - \xv^{(k)}\|_2^2 - (\vv^{(k)} - \xv^{(k)})\ts \Hmat\ts \Rmat\inv \Hmat (\vv^{(k)} - \xv^{(k)}) \nonumber\\
        &= \|\vv^{(k)} - \xv^{(k)}\|_2^2- \|\Rmat^{-\frac{1}{2}}\Hmat (\vv^{(k)} - \xv^{(k)})\|_2^2. \nonumber
    \end{align}
	In \eqref{Eq:convg_C} we have used the bounded denoiser.
	Using Assumption \ref{Ass:non_in} (non-increasing denoiser), we have
	$\sigma_k \rightarrow 0$,  $\|\xv^{(k+1)} - \xv^{(k)}\|_2^2 \rightarrow 0$ and thus $\xv^{(k)}$ converges.
\end{proof}

\subsection{PnP-ADMM vs. PnP-GAP}
Comparing PnP-GAP in Eqs~\eqref{Eq:x_k+1} and \eqref{Eq:Denoise_GAP} and PnP-ADMM in Eqs~\eqref{Eq:solvex}-\eqref{Eq:u_k+1}, we can see that PnP-GAP only has two subproblems (rather than three as in PnP-ADMM) and thus the computation is faster. 
It was pointed out in~\cite{Liu18TPAMI} that in the noise-free case, ADMM and GAP perform the same with appropriate parameter settings, which has been mathematically proved. However, in the noisy case, ADMM usually performs better since it considers noise in the model and below we give a geometrical explanation.

\begin{figure}[htbp!]
	\begin{center}
		\includegraphics[width=\linewidth]{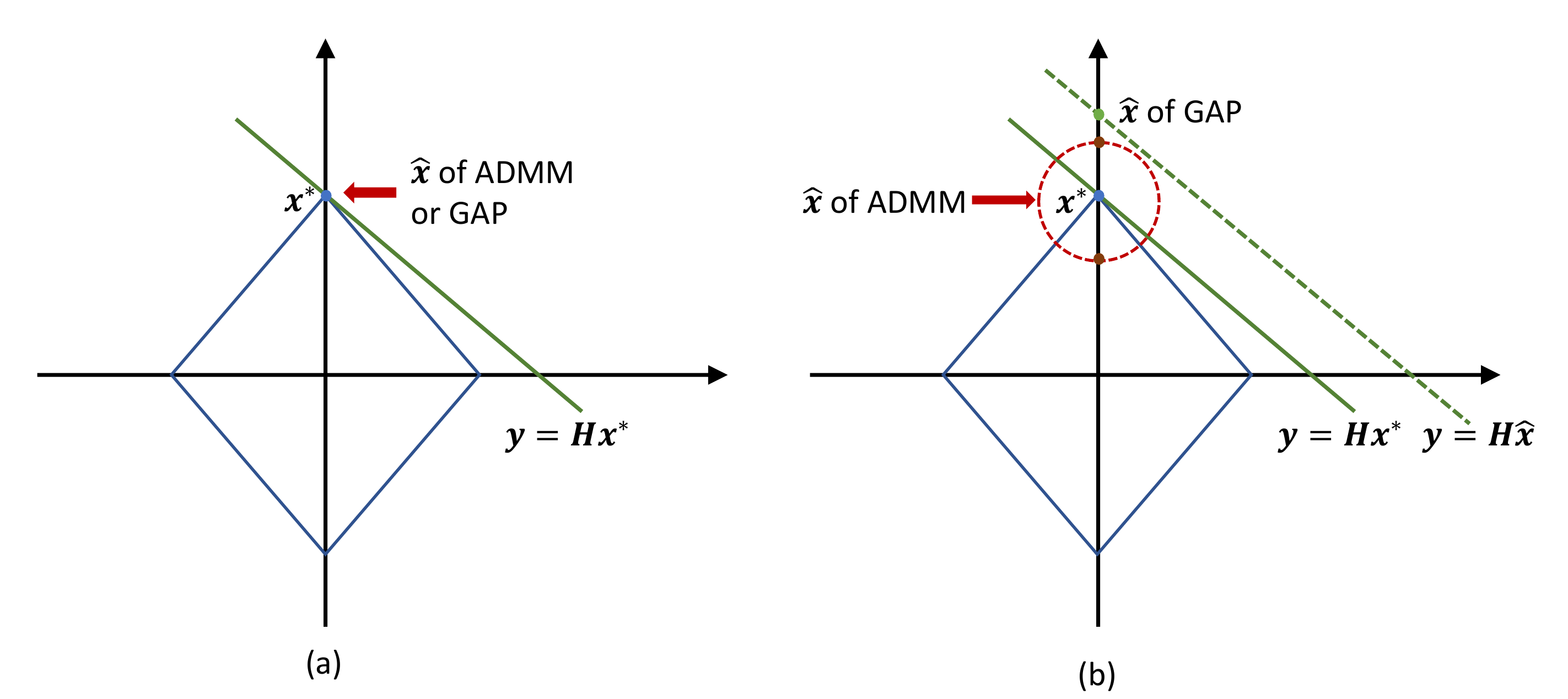}
	\end{center}
	\vspace{-3mm}
	\caption{Demonstration of the solution of ADMM and GAP for a two-dimensional sparse signal, where $\xv^*$ denotes the truth.
	(a) In the noise-free case, both ADMM and GAP have a large chance to converge to the true signal. (b) In the noisy case,  GAP will converge to the green dot (the cross-point of dash-green line and vertical axis), whereas the solution of ADMM will be one of the two red dots that the red circle crosses the vertical axis.} 
	\label{fig:ADMM_GAP}
\end{figure}
 
In Fig.~\ref{fig:ADMM_GAP}, we use a two-dimensional sparse signal (with the $\ell_1$ assumption shown as the diamond shape in blue lines in Fig.~\ref{fig:ADMM_GAP}) as an example to compare ADMM and GAP. Note that the key difference is that, in both noise-free and noisy cases, GAP always imposes the solution $\hat\xv$ on the line of $\yv = \Hmat\hat\xv$ by Eq.~\eqref{Eq:x_k+1}.
In the noise-free case in Fig.~\ref{fig:ADMM_GAP}(a), we can see that since GAP imposes $\hat\xv$ on the green line, it will converge to the true signal $\xv^*$.
ADMM does not have this constraint but minimizes $\|\yv-\Hmat\xv\|_2^2$, the solution might be a little bit off the true signal $\xv^*$. However, with appropriate parameter settings and a good initialization, it also has a large chance to converge to the true signal. 
In the noisy case, GAP sill imposes $\hat\xv$ on the line of $\yv = \Hmat\hat\xv$, shown by the dash-green line in Fig.~\ref{fig:ADMM_GAP}(b). In this case, due to noise, this line might deviate from the solid green line where the true signal lies on. GAP will thus converge to the green-point where the dash-green line crosses the vertical axis. 
On the other hand, by minimizing $\|\yv-\Hmat\xv\|_2^2$, the solution of ADMM can be in the dash-red circle depending on the initialization. 
Considering the sparse constraint, the final solution of ADMM would be one of the two red dots that the red circle crosses the vertical axis.
Therefore, in the noisy case, the Euclidean distance between the GAP solution and the true signal ($\|\hat{\xv} - \xv^*\|_2$) might be larger than that of ADMM. However, the final solution of ADMM depends on the initialization and it is not guaranteed to be more accurate than GAP.

The PnP framework can be recognized as a deep denoising network plus an inverse problem solver. Other solvers such as TwIST~\cite{Bioucas-Dias2007TwIST} and FISTA~\cite{Beck09IST} can also be used~\cite{Zheng20_PRJ_PnP-CASSI} and may also lead to convergence results under proper conditions.
According to our experience, TwIST usually converges slowly and FITSA sometimes sticks to limited performance.
Hence, in the experiments, we use PnP-GAP for simulation data and PnP-ADMM for real data.

\section{Integrate Various Denoisers into PnP for SCI Reconstruction\label{Sec:P3}}
%

In the above derivation, we assume the denoiser existing and in this section, we briefly introduce different denoisers. These denoisers have different speed and quality.

\subsection{Non-deep Denoiser}
In conventional denoising algorithms, a prior is usually employed to impose the piece-wise constant (by TV), sparsity (by bases or learn-able dictionaries) or low-rank (similar patches groups).
These algorithms usually have a clear objective function. In the following, we briefly categorize these algorithms into following classes. For a detailed review, please refer to~\cite{Zha2020TIP_JPG,Zha2020TIP_NSSP}.
\begin{list}{\labelitemi}{\leftmargin=8pt \topsep=2pt \parsep=1pt}
    \item Global constraint based algorithms such as TV~\cite{Rudin92_TV,Stanley05_TV} minimize the total variations of the entire image and have been extended to videos~\cite{yang2013efficient}.
    \item Global sparsity based algorithms impose the coefficients of the image under specific basis such as wavelet~\cite{Crouse98_wavelet} or Curvelet~\cite{Starck02_Curvelet}.
    \item Patch based algorithms usually learn a dictionary using methods such as K-SVD~\cite{Aharon06TSP} for image patches and  then impose sparsity on the coefficients.
    \item Patch-group based algorithms exploit the nonlocal similarity of image patches and impose the sparsity~\cite{Mairal_09ICCV_LSSC} or low-rank~\cite{Gu14CVPR,Gu17IJCV} on these similar patch groups.
\end{list}

Among these denoisers, usually, a faster denoiser~\eg, TV, is very efficient, but cannot provide high-quality results.  
The middle class algorithms~\eg, K-SVD and BM3D~\cite{Dabov07BM3D} can provide decent results with a longer running time.
More advanced denoising algorithm such as WNNM~\cite{Gu14CVPR,Gu17IJCV} can provide better results, but even slower.
On the other hand, while extensive denoising algorithms for images have been developed, VBM4D~\cite{Maggioni2012VideoDD} is still one of the state-of-the-art algorithms for video denoising. 
In our previous work, we have extended WNNM into SCI for gray-scale videos leading to the state-of-the-art algorithm (DeSCI~\cite{Liu18TPAMI}) for SCI, which performs better than PnP-VBM4D as shown in Fig.~\ref{fig:demo} but paying the price of a longer running time.

\subsection{Deep Denoiser}
Another line of emerging denoising approaches is based on deep learning~\cite{XieNIPS2012_deepDN,zhang2017beyond}, which can provide decent results within a short time after training, but they are usually not robust to noise levels and in high noisy cases, the results are not good.
Since this paper is not focusing directly on video denoiser, we do not provide a detailed survey on the deep learning based video denoising. Interested readers can refer to other recent papers.

Different from conventional denoising problems, in SCI reconstruction, the noise level in each iteration is usually from large to small and the dynamic range can from 150 to 1, considering the pixel values within $\{0,1,\dots, 255\}$. 
Therefore, a flexible denoiser that is robust to the input noise level is desired. 
Fortunately, FFDNet~\cite{Zhang18TIP_FFDNet} has provided us a fast and flexible solution under various noise levels. 
However, since FFDNet is developed for images, we perform the denoising step in PnP for SCI frame-wise; for the color SCI problem, we used the gray-scale denoising for each channel in~\cite{Yuan20CVPR}.
As discussed before and shown in Fig.~\ref{fig:Bayer_sci}, using joint demosaicing and reconstruction by employing color denoiser of FFDNet instead of grayscale ones can improve the results significantly. Please refer to Table~\ref{Tab:results_midscale} and Fig.~\ref{fig:comp_frames_midscale} .

Most recently, we have noticed that the FastDVDnet~\cite{Tassano_2020_CVPR} also satisfies these desired (fast and flexible) properties. More importantly, FastDVDnet is developed for video denoising which takes account of the strong temporal correlation within consequent video frames. By using FastDVDnet into PnP, we have achieved even better results on both grayscale and color videos than those of FFDNet. 
{Please refer to Table~\ref{Tab:results_4video} for grayscale video SCI results and Table~\ref{Tab:results_midscale} for the color SCI}.

\subsection{Hybrid Denoiser}
By integrating these denoising algorithms into PnP-GAP/ADMM, we can have different algorithms (Table~\ref{Tab:results_4video} and Fig.~\ref{fig:demo}) with  different results. It is worth noting that DeSCI can be seen as PnP-WNNM, and its best results are achieved by exploiting the correlation across different video frames. 
On the other hand, most existing deep denoising priors are still based on images. Therefore, it is not unexpected that the results of PnP-GAP/ADMM-FFDNet are not as good as DeSCI. 
As mentioned above, by using video denoising priors such as FastDVDnet, the results can be improved.

In addition, these different denoisers can be used in parallel, \ie, one after each other in one GAP/ADMM iteration or used sequentially, \ie, the first $K_1$ iterations using FFDNet and the next $K_2$ iterations using WNNM to achieve better results. This is a good way to balance the performance and running time. Please refer to the performance and running time in Fig.~\ref{fig:demo} and {Table~\ref{Tab:results_4video}}. 
These different denoising priors can also be served as the complementary priors in image/video denoising~\cite{zha2020power,Qiao2020_APLP}. 

\subsection{Joint Demosaicing and SCI Reconstruction \label{Sec:jointcsci}}


For the color SCI described in Eq.~\eqref{eq:forward_color}, though it is easy to directly use the PnP algorithm derived above for grayscale videos by changing the forward matrix, it cannot leads to good results by using the deep denoiser such as FFDNet or FastDVDnet only for denoising according to our experiments.
This might due to the fact that video SCI can be recognized as temporal interpolation while the demosaicing is spatial interpolation. Jointly performing these two tasks is too challenging for the deep denoiser, especially for videos~\cite{GharbiACM16}.
To cope with this challenge, we rewrite the forward model of color-SCI as
\begin{equation}
  \yv = \Hmat \Tmat_{\cal M} \xv  
\end{equation}
where $\Tmat_{\cal M} $ is the mosaicing and deinterleaving process shown in Fig.~\ref{fig:Bayer_sci} that translate the RGB videos $\xv$ to the Bayer pattern video $\xv^{\rm (rggb)}$. 
From $\yv$ to $\xv^{\rm (rggb)}$ is exactly the same as grayscale video SCI. 

Different algorithms exist for images demosaicing, \ie, from $\xv^{\rm (rggb)}$ to $\xv$~\cite{LiDemosaicing08}. Recently, the deep learning based demosaicing algorithms have also been developed~\cite{Brady:20,GharbiACM16}.  
The key principle of PnP algorithms for color SCI is to make full use of the color denoising algorithm, rather than channel-wise grayscale denoising. Bearing this concern in mind, we decompose the denoising step in PnP-GAP \eqref{Eq:Denoise_GAP} into the following 2 steps:
\begin{eqnarray}
    \tilde{\vv}^{k+1} &=& {\cal D}_{\cal M} (\xv^{(k+1)}), \label{Eq:demosaic}\\
    {\vv}^{k+1} &=& {\cal D}_{\sigma}(\tilde{\vv}^{k+1}), \label{Eq:denoise_color}
\end{eqnarray}
where ${\cal D}_{\cal M}$ is the demosaicing algorithm being used\footnote{Demosaicing is conducted after interleaving. Similarly, mosaicing and then deinterleaving are conducted before projection in Eq.~\eqref{Eq:x_k+1}, as shown in Fig.~\ref{fig:Bayer_sci}(b).} and ${\cal D}_{\sigma}$ now denotes the color denoising algorithms. 
Note that ${\cal D}_{\cal M}$ is also a plug-and-play operation where different algorithms can be used. 
In the experiments, we have found that both 'Malvar04'~\cite{malvar2004high-quality} and `Menon07'~\cite{Menon07} can lead to stable results and they are performing better than the fast bilinear interpolation.
However, the time consumption of `Menon07' is about 4$\times$ longer than 'Malvar04' with limited gain for our color SCI problem. Therefore, we used 'Malvar04' in our experiments.
We believe a deep learning based algorithm can leads to better results for demosaicing. However, for our color-SCI reconstruction, we noticed that existing deep learning based demosaicing is not stable, though they do perform well for the sole demosaicing task.
We leave this robust demosaicing deep learning network for the future work.
Nonetheless, our proposed PnP algorithm can adopt any new demosaicing and denoising algorithms to improve the results. 

For the sake of running time of large-scale color SCI reconstruction, instead of calling the  demosaicing algorithms in Eq.~\eqref{Eq:demosaic}, we have also developed a light-weight method by using the R channel, B channel and the averaged G channels to construct a small-scale RGB video (with half size of rows and columns of the desired video) in each iteration; following this, the color denoising operation in~\eqref{Eq:denoise_color} is performed on this small scale RGB video and the iterations are conducted until converge. 
In this case, we only need to call the demosaicing algorithm once at the end to provide the final result.
Apparently, the results would not be as good as using demosaicing in each iteration, but it saves time.  


\subsection{Online PnP for Sequential Measurements}
To speed up the convergence of PnP with deep denoiser, we notice that a good initialization would help.
Usually, we use a few iterations of GAP-TV to warmly start the PnP-FastDVDnet and it will lead to good results with about 50 iterations.
However, for the sequential measurements, we empirically find that using the reconstruction results of the previous measurement to initialize the next measurement will lead to good results. 
In this case, when we are handling multiple sequential measurements, we only need a warm start for the first measurement and the next ones can use the results of previous measurements. 
This will slightly save the reconstruction time.
This idea shares the similar spirit of `group of pictures' in the MPEG video compression.

\begin{table*}[!htbp]
	\caption{Grayscale benchmark dataset: The average results of PSNR in dB (left entry in each cell) and SSIM (right entry in each cell) and run time per measurement/shot in minutes by different algorithms on 6 benchmark datasets.}
	\centering
	 \vspace{-3mm}
	\resizebox{\textwidth}{!}
	{
	\begin{threeparttable}
	\begin{tabular}{cV{2}ccccccV{2}cc}
			\hlineB{3}
			Algorithm& \texttt{Kobe} & \texttt{Traffic} & \texttt{Runner} & \texttt{Drop} & \texttt{Crash} & \texttt{Aerial} & Average &  Run time (min) \\
			\hlineB{3}
			GAP-TV~\cite{Yuan16ICIP_GAP}   & 26.46, 0.8448 & 20.89, 0.7148 & 28.52, 0.9092 & 34.63, 0.9704 & 24.82, 0.8383 & 25.05, 0.8281 & 26.73, 0.8509 & 0.07 \\
			{DeSCI~\cite{Liu18TPAMI}} & {\bf 33.25}, {0.9518} & {\bf 28.71}, {0.9250} & {\bf 38.48}, {\bf 0.9693} & 43.10, 0.9925 & {27.04}, {0.9094} & 25.33, 0.8603 & {\bf 32.65}, {0.9347} & 103.0 \\
			\hlineB{2}
			PnP-VBM4D        & 30.60, 0.9260 & 26.60, 0.8958 & 30.10, 0.9271 & 26.58, 0.8777 & 25.30, 0.8502 & 26.89, 0.8521 & 27.68, 0.8882 & 7.9  \\
			
			PnP-FFDNet~\cite{Yuan20CVPR}       & 30.50, 0.9256 & 24.18, 0.8279 & 32.15, 0.9332 & 40.70, 0.9892 & 25.42, 0.8493 & 25.27, 0.8291 & 29.70, 0.8924 & {0.05 (GPU)} \\
			
			PnP-WNNM-TV      & 33.00, 0.9520 & 26.76, 0.9035 & 38.00, 0.9690 & 43.27, 0.9927 & 26.25, 0.8972 & 25.53, 0.8595 & 32.14, 0.9290 & 40.8 \\
			PnP-WNNM-VBM4D   & 33.08, \bf{0.9537} & 28.05, 0.9191 & 33.73, 0.9632 & 28.82, 0.9289 & 26.56, 0.8874 & {27.74}, {0.8852} & 29.66, 0.9229 & 25.0 \\
			PnP-WNNM-FFDNet & 32.54, 0.9511 & 26.00, 0.8861 & 36.31, 0.9664 & \bf{43.45}, \bf{0.9930} & 26.21, 0.8930 & 25.83, 0.8618 & 31.72, 0.9252 & 17.9 \\
			\hlineB{2}
			GAP-TV*          & 26.92, 0.8378 & 20.66, 0.6905 & 29.81, 0.8949 & 34.95, 0.9664 & 24.48, 0.7988 & 24.81, 0.8105 & 26.94, 0.8332 & {\bf 0.03} \\
			PnP-FFDNet*      & 30.33, 0.9252 & 24.01, 0.8353 & 32.44, 0.9313 & 39.68, 0.9864 & 24.67, 0.8330 & 24.29, 0.8198 & 29.21, 0.8876 & {{\bf 0.03} (GPU)} \\
			\rowcolor{lightgray}
			PnP-FastDVDnet*  & 32.73, 0.9466 & 27.95, {\bf 0.9321} & 36.29, 0.9619 & 41.82, 0.9892 & {\bf 27.32}, {\bf 0.9253} & {\bf 27.98}, {\bf 0.8966} & 32.35, {\bf 0.9420} & {0.10 (GPU)} \\
			\hlineB{3}
	\end{tabular}
	\begin{tablenotes}
	  \item[*] Implemented with Python (PyTorch for FFDNet and FastDVDnet), where the rest are implemented with MATLAB (MatConvNet for FFDNet).
	  \end{tablenotes}
	  \end{threeparttable}
	}
	\label{Tab:results_4video}
\end{table*}

\section{Simulation Results \label{Sec:results}}
We apply the proposed PnP algorithms to both simulation~\cite{Liu18TPAMI,Ma19ICCV} and real datasets captured by the SCI cameras~\cite{Patrick13OE,Yuan14CVPR,Qiao2020_APLP}.
In addition to the widely used grayscale benchmark datasets~\cite{Yuan20CVPR}, we also build a mid-scale color dataset consisting of 6 color videos (details in Sec.~\ref{sec:sim_color}) and we hope they will serve as the benchmark data of color SCI problems. 
This benchmark dataset is used to verify the performance of our proposed PnP-GAP for joint reconstruction and demosaicing compared with other algorithms. 
Finally, we apply the proposed joint method to the large-scale datasets introduced in~\cite{Yuan20CVPR} and show better results using FastDVDnet for color video denoising.

Conventional denoising algorithms include TV~\cite{Yuan16ICIP_GAP}, VBM4D~\cite{Maggioni2012VideoDD} and WNNM~\cite{Gu14CVPR} are used for comparison. For the deep learning based denoiser, we have tried various algorithms and found that FFDNet~\cite{Zhang18TIP_FFDNet} provides the best results among image denoising methods while FastDVDnet~\cite{Tassano_2020_CVPR}  provides the best results among video denoising approaches.
Both PSNR and SSIM~\cite{Wang04imagequality} are employed as metrics to compare different algorithms. 
Note that in our preliminary paper~\cite{Yuan20CVPR}, all the codes are conducted in MATLAB, while in this work, we re-write the code of GAP-TV, PnP-FFDNet using Python, to be consistent with PnP-FastDVDnet. However, we do notice a difference of FFDNet; the performance of Python version\footnote{Code from \url{https://github.com/cszn/KAIR}.} is a little bit worse (0.49dB lower in PSNR for the grayscale benchmark data) than the counterpart conducted in MATLAB\footnote{Code from \url{https://github.com/cszn/FFDNet}.}. 
We also notice that GAP-TV in Python is more than $2\times$ faster than MATLAB with a slightly better result. We show these in Table~\ref{Tab:results_4video}.

\begin{figure}[htbp!]
	\begin{center}
		\includegraphics[width=1.0\linewidth]{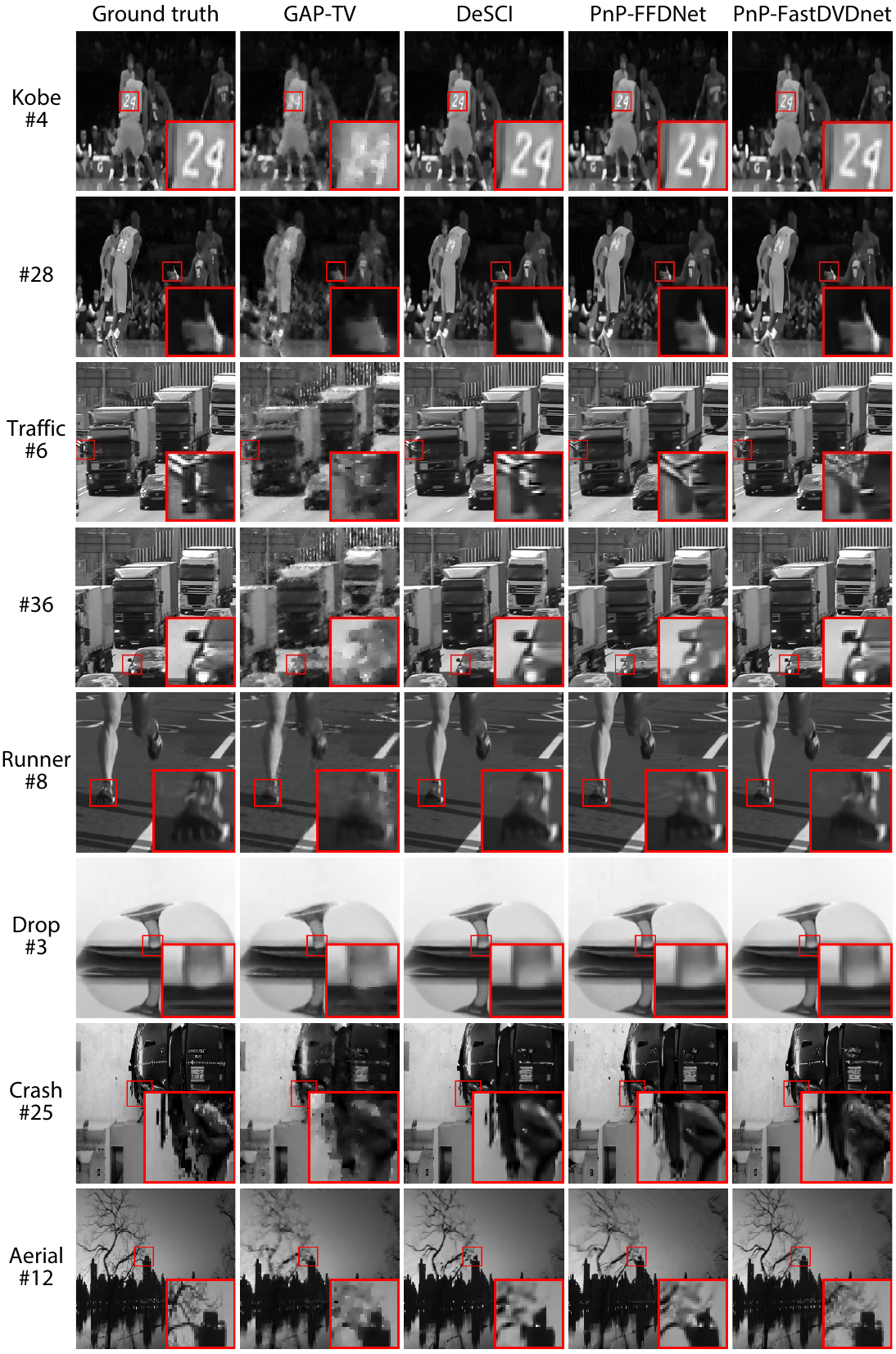}
	\end{center}
	\vspace{-4mm}
	\caption{Comparison of reconstructed frames of different PnP-GAP algorithms (GAP-TV~\cite{Yuan16ICIP_GAP}, DeSCI~\cite{Liu18TPAMI}, PnP-FFDNet~\cite{Yuan20CVPR}, and PnP-FastDVDnet) on six simulated grayscale video SCI datasets of spatial size $256\times256$ and $B=8$.}
	\label{fig:comp_frames_full}
	\vspace{-2mm}
\end{figure}

\subsection{Benchmark Data: Grayscale Videos \label{sec:sim_gray}}
We follow the simulation setup in~\cite{Liu18TPAMI} of six datasets, \ie, \texttt{Kobe, Traffic, Runner, Drop, crash,} and \texttt{aerial}~\cite{Ma19ICCV}\footnote{The results of DeSCI (GAP-WNNM) is different from those reported in \cite{Ma19ICCV} because of parameter settings of DeSCI, specifically the input estimated noise levels for each iteration stage. We use exactly the same parameters as the DeSCI paper~\cite{Liu18TPAMI}, which is publicly available at \href{https://github.com/liuyang12/DeSCI}{https://github.com/liuyang12/DeSCI}.}, where $B=8$ video frames are compressed into a single measurement.
Table~\ref{Tab:results_4video} summarizes the PSNR and SSIM results of these 6 benchmark data using various denoising algorithms, where DeSCI can be categorized as GAP-WNNM, and PnP-WNNM-FFDNet used 50 iterations FFDNet and then 60 iterations WNNM, similar for PnP-WNNM-VBM4D.
{PnP-FastDVDnet used 60 iterations and we used 5 neighbouring frames for video denoising.}
It can be observed that:
\begin{list}{\labelitemi}{\leftmargin=8pt \topsep=2pt \parsep=1pt}
	\item [$i$)] By using GPU, PnP-FFDNet is now the fastest algorithm\footnote{{Only a regular GPU is needed to run FFDNet and since FFDNet is performed in a frame-wise manner, we do not need a large amount of CPU or GPU RAM (no more than 2GB here) compared to other video denoisers using parallelization (even with parallelization, other algorithms listed here are unlikely to outperform PnP-FFDNet in terms of speed).}}; it is very close to GAP-TV, meanwhile providing more than 2dB higher PSNR than GAP-TV. Therefore, PnP-FFDNet can be used as {\em an efficient baseline} in SCI reconstruction. Since the average PSNR is close to 30dB, it is applicable in real cases. This will be further verified in the following subsections on mid-scale and large-scale color datasets.
	\item [$ii$)]
	DeSCI still provides the best results on average PSNR; however, by combing other algorithms with WNNM, comparable results (\eg PnP-WNNM-FFDNet) can be achieved by only using $1/6$ computational time.
	\item [$iii$)] {PnP-FastDVDnet provides the best results on average SSIM and regarding PSNR, it is only 0.3dB lower than DeSCI but $1000\times$ faster. PnP-FastDVDnet is the only algorithm that can provide a higher SSIM than DeSCI.}
	\item [$iv$)] Comparing PnP-FFDnet with PnP-FastDVDnet, we observe that utilizing the temporal correlation has improved the results significantly, \ie, more than 3dB in PSNR and 0.05 in SSIM. 
\end{list}

Fig.~\ref{fig:comp_frames_full} plots selected frames of the six datasets using different algorithms. It can be seen that though DeSCI still leads to highest PSNR, the difference between PnP-FastDVDNet and DeSCI is very small and in most cases, they are close to each other. 
By investigating the temporal correlation, PnP-FastDVDNet can provide finer details than PnP-FFDNet and sometimes DeSCI; please refer to the zoomed parts of \texttt{Aerial} in Fig.~\ref{fig:comp_frames_full}.

\begin{table*}[!htbp]
	\caption{Mid-scale Bayer benchmark dataset: the average results of PSNR in dB (left entry in each cell) and SSIM (right entry) and running time per measurement/shot in minutes by different algorithms on 6 benchmark color Bayer datasets. All codes are implemented in Python (Pytorch for deep denoising) except DeSCI, which is using MATLAB.} 
	\centering
	 \vspace{-4mm}
	\resizebox{.955\textwidth}{!}
	{
	\begin{tabular}{cV{2}ccccccV{2}cc}
			\hlineB{3}
			Algorithm& \texttt{Beauty} & \texttt{Bosphorus} & \texttt{Jockey} & \texttt{Runner}  & \texttt{ShakeNDry}& \texttt{Traffic} & Average &  Run time (mins) \\
			\hlineB{3}
			GAP-TV           & 33.08, 0.9639 & 29.70, 0.9144 & 29.48, 0.8874 & 29.10, 0.8780 & 29.59, 0.8928 & 19.84, 0.6448   &28.46, 0.8635   &   0.3\\
			{DeSCI (GAP-WNNM)} & {34.66}, {0.9711} & 32.88, 0.9518 & 34.14, 0.9382 & 36.16, 0.9489 & 30.94, 0.9049 &24.62, 0.8387   & 32.23, 0.9256  &  1544  \\
			\hline
			PnP-FFDNet-gray       & 33.21, 0.9629 & 28.43, 0.9046 & 32.30, 0.9182 & 30.83, 0.8875 & 27.87, 0.8606 & 21.03, 0.7113  & 28.93, 0.8742  & { 0.22 (GPU)} \\
			PnP-FFDNet-color      & 34.15, 0.9670 & 33.06, 0.9569 & 34.80, 0.9432 & 35.32, 0.9398 & 32.37, 0.9401  &24.55, 0.8370  & 32.38, 0.9307  &   1.63 (GPU)\\
 			\hline
			PnP-FastDVDnet-gray  & 33.01, 0.9628 & 30.95, 0.9342 & 33.51, 0.9279 & 32.82, 0.9004 & 29.92, 0.8920 & 22.81, 0.7764  &  30.50, 0.8989 & { 0.33 (GPU)} \\
			\rowcolor{lightgray}
			PnP-FastDVDnet-color  & {\bf 35.27}, {\bf 0.9719} & {\bf 37.24}, {\bf 0.9781} & {\bf 35.63}, {\bf 0.9495} & {\bf 38.22}, {\bf 0.9648} & {\bf 33.71}, {\bf 0.9685} & {\bf 27.49}, {\bf 0.9147}  & {\bf 34.60}, {\bf 0.9546}  & {1.65 (GPU)} \\
			\hlineB{3}
	\end{tabular}
	}
	\label{Tab:results_midscale}
\end{table*}

\begin{figure*}[htbp!]
	\begin{center}
		\includegraphics[width=.955\linewidth]{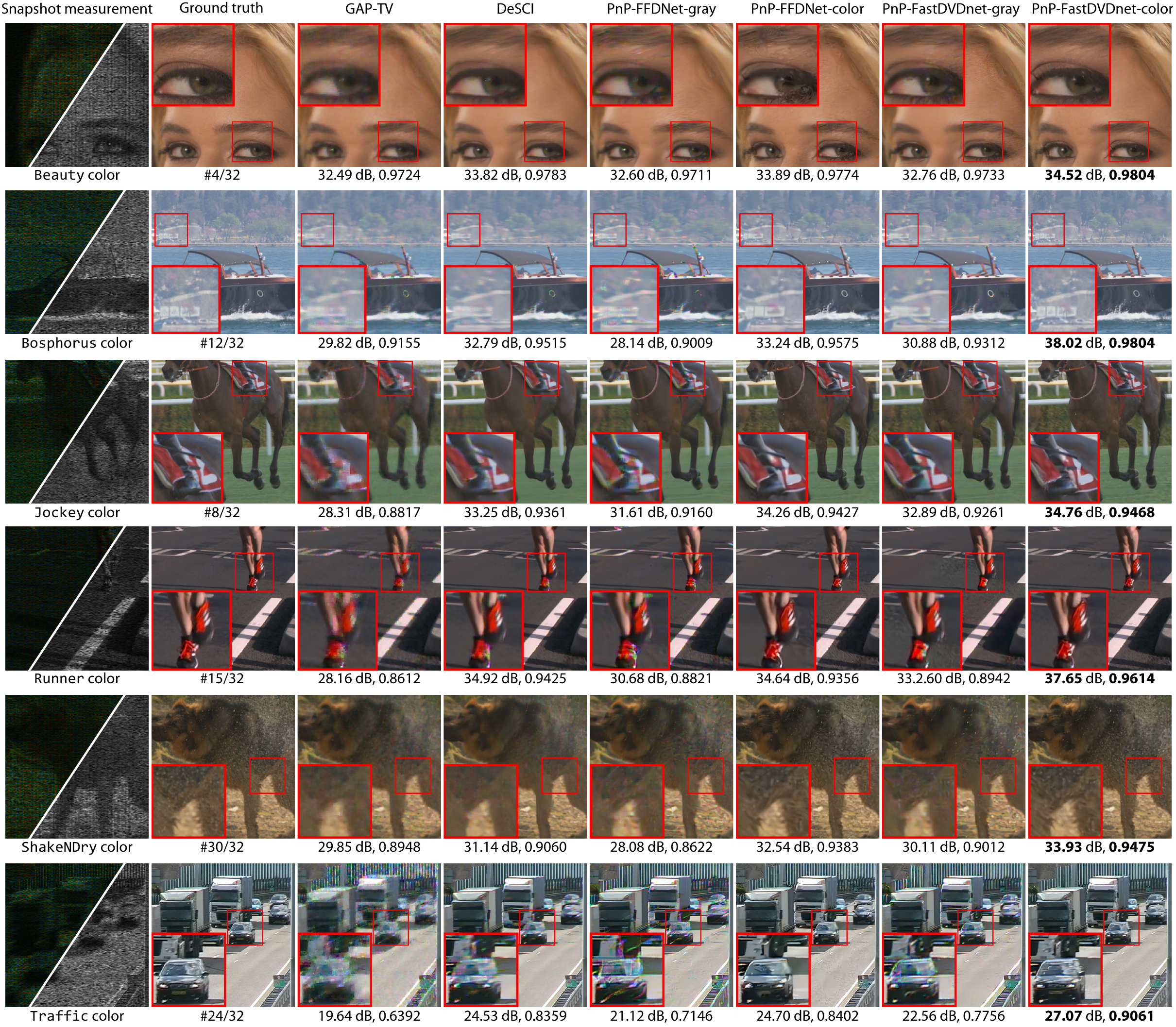}
	\end{center}
	\vspace{-4mm}
	\caption{Comparison of reconstructed frames of PnP-GAP algorithms (GAP-TV~\cite{Yuan16ICIP_GAP}, DeSCI~\cite{Liu18TPAMI}, PnP-FFDNet~\cite{Yuan20CVPR}, and PnP-FastDVDnet) on six simulated benchmark color video SCI datasets of size $512\times512\times3$ and $B=8$. 
	Please refer to the full videos in the supplementary material. 
	}
	\label{fig:comp_frames_midscale}
	\vspace{-3mm}
\end{figure*}

\begin{figure*}[htbp!]
	\centering
		\includegraphics[width=1.0\linewidth]{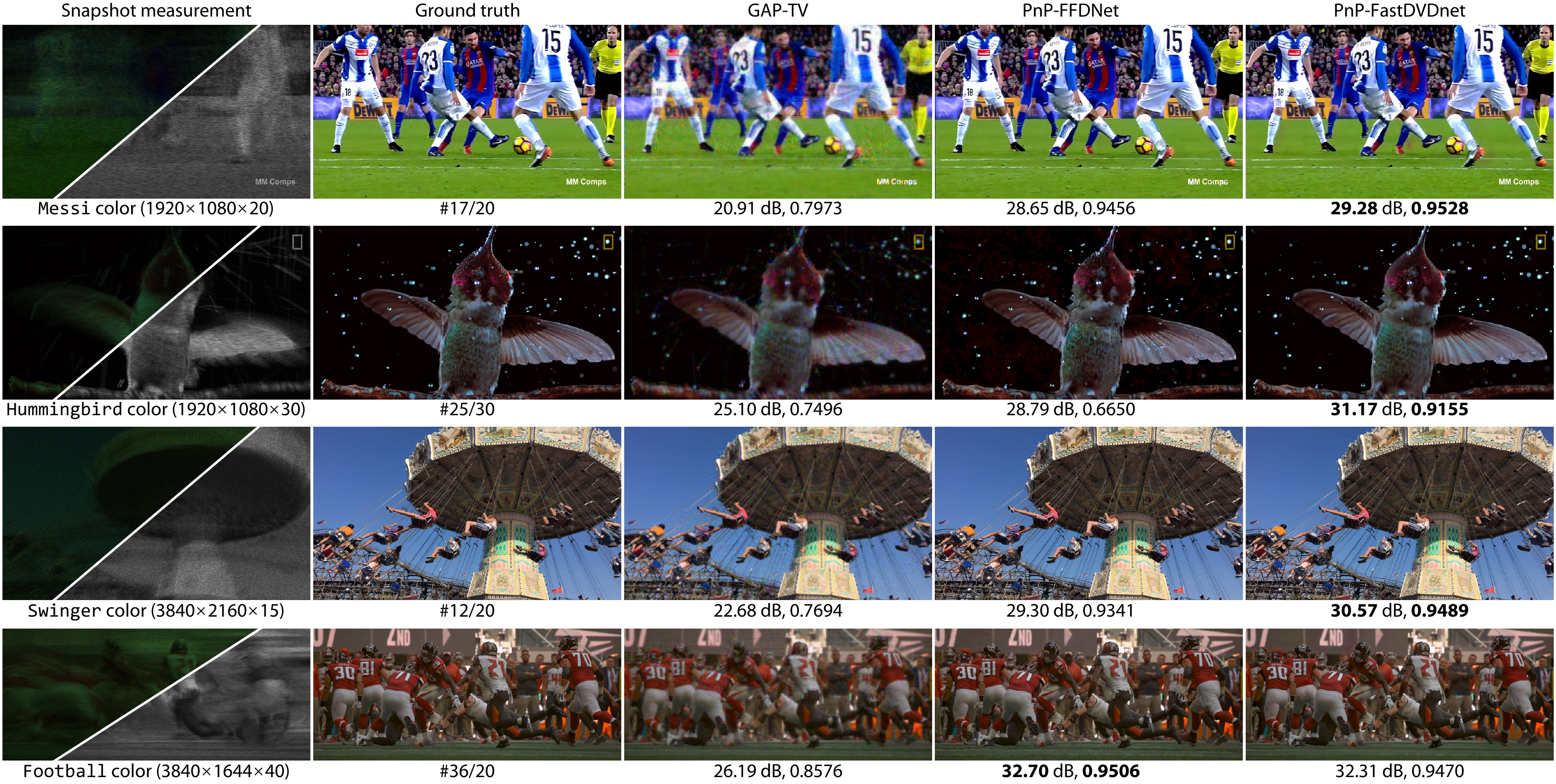}
	\vspace{-6mm}
	\caption{Reconstructed frames of PnP-GAP algorithms (GAP-TV~\cite{Yuan16ICIP_GAP},  PnP-FFDNet~\cite{Yuan20CVPR}, and PnP-FastDVDnet) on four simulated large-scale video SCI datasets. 
	Please refer to the full videos in the supplementary material.
	}
	\vspace{-3mm}
	\label{fig:comp_largescale}
\end{figure*}

\subsection{Benchmark Data: Color RGB-Bayer Videos \label{sec:sim_color}}
As mentioned before, in this paper, we propose the joint reconstruction and demosaicing using the PnP framework for color SCI shown in the lower part of Fig.~\ref{fig:Bayer_sci}. To verify the performance, we hereby generate a color RGB video dataset with 6 scenes of spatial size $512\times512\times3$, and here $3$ denotes the RGB channels. Similar to the grayscale case, we used compression rate $B=8$.
The schematic of a color video SCI system is shown in Fig.~\ref{fig:video_color_sci}. Every 8 consequent video frames are first interleaved to the mosaic frames of size $512\times 512 \times 8$; then these mosaic frames are modulated by shifting binary masks of size $512\times 512 \times 8$ and finally summed to get the compressed mosaic measurement of size $512\times 512$.  
For each dataset, we have 4 compressed measurements and thus in total 32 RGB video frames.
As shown in Fig.~\ref{fig:comp_frames_midscale}, these datasets include \texttt{Beauty}, \texttt{Bosphorus}, \texttt{Jockey}, \texttt{ShakeNDry}\footnote{\texttt{Beauty}, \texttt{Bosphorus}, \texttt{Jockey}, \texttt{ShakeNDry} are downloaded from \href{http://ultravideo.cs.tut.fi/\#testsequences}{http://ultravideo.cs.tut.fi/\#testsequences}.}, \texttt{Runner}\footnote{{Downloaded from \href{https://www.videvo.net/video/elite-runner-slow-motion/4541}{https://www.videvo.net/video/elite-runner-slow-motion/4541}.}}
and \texttt{Traffic}\footnote{Downloaded from \href{http://dyntex.univ-lr.fr/database.html}{http://dyntex.univ-lr.fr/database.html}.}.
In order to keep the video quality, we crop (instead of resize) the video frames with a spatial size of $512\times 512$. We dub these datasets the `mid-scale color data' due to the size in-between the small size grayscale benchmark data and the large-scale data discussed in the next subsection. 

For other algorithms, we perform the reconstruction and demosaicing separately.
The R, G1, G2, and B channels are reconstructed separately and then we employ the `\texttt{demosaic}' function in MATLAB to get the final RGB video.
To verify the performance of this joint procedure, we compare the PnP-FFDnet/FastDVDnet using joint processing (color denoising) and separately reconstruction (grayscale denoising) shown in Fig.~\ref{fig:Bayer_sci}.  

{Table~\ref{Tab:results_midscale} summarizes the PSNR and SSIM results of these datasets. 
We have the following observations.}
\begin{list}{\labelitemi}{\leftmargin=8pt \topsep=2pt \parsep=1pt}
    \item [$i)$] When using color denoising for joint reconstruction and demosaicing, both PnP-FFDNet-color and PnP-FastDVDnet-color outperform DeSCI.
    \item [$ii)$] Color denoising significantly improves the results over grayscale denoising, \ie, for FFDNet, the improvement is 3.45dB and for FastDVDnet, it is even 4.1dB in PSNR.
    \item [$iii)$] Regarding the running time, both PnP-FastDVDnet and PnP-FFDNet need about 1.6 minutes per measurement, while they only need about 0.3 minutes for their grayscale counterparts. Therefore, most time was consumed by demosaicing. As mentioned in Sec.~\ref{Sec:jointcsci}, we hope deep learning based demosaicing will provide a fast and better result in the future.
\end{list}

\begin{table*}[htbp!]
\caption{Running time (minutes) of large-scale data using different algorithms.} 
\vspace{-3mm}
\resizebox{1\linewidth}{!} {
\begin{tabular}{ccV{3}ccccc}
\hlineB{3}
Large-scale dataset & Pixel resolution         & GAP-TV & PnP-FFDNet-gray & PnP-FFDNet-color & PnP-FastDVDnet-gray & PnP-FastDVDnet-color \\
\hlineB{3}
{\tt Messi} color      & $1920\times1080\times3\times20$ & 15.1   & 5.2             & 42.2             & 7.6                 & 43.1                 \\
{\tt Hummingbird} color & $1920\times1080\times3\times30$ & 20.3   & 6.6             & 61.2             & 10.6                & 54.0                 \\
{\tt Swinger} color    & $3840\times2160\times3\times15$ & 39.2   & 13.2            & 138.8            & 21.3                & 138.4                \\
{\tt Football} color   & $3840\times1644\times3\times40$ & 83.0   & 30.6            & 308.8            & 50.7                & 298.1                \\
\hlineB{3}
\end{tabular}
}
\label{Table:time_largescale}
\end{table*}

{Figure~\ref{fig:comp_frames_midscale} plots selected  reconstruction frames of different algorithms for these 6 RGB Bayer datasets with the snapshot measurement shown on the far left. 
Note that, due to the Bayer pattern of the sensor, the captured measurement is actually grayscale shown on the lower-right, whereas due to the coding in the imaging system, the demosaiced measurement depicts the wrong color shown in the upper-left part of the measurement. 
It can be seen from Fig.~\ref{fig:comp_frames_midscale} that PnP-FastDVDnet-color and PnP-FFDNet-color are providing smooth motions and fine spatial details. Some color mismatch exists in the GAP-TV, DeSCI and PnP-FFDNet-gray and PnP-FastDVDnet-gray. For instance, in the {\texttt{Traffic}} data, the color of the cars is incorrectly reconstructed for these methods; similar case exists in the water drops in the {\texttt{ShakeNDry}}. Overall, GAP-TV provides blurry results and DeSCI sometimes over-smooths the background such as the lawn in the {\texttt{Jockey}} data. 
PnP-FastDVDnet-color provides the finest details in the complicated background such as the trees in \texttt{Bosphorus}. }

\subsection{Large-scale Data}
Similar to the benchmark data, we simulate the color video SCI measurements for large-scale data with four YouTube slow-motion videos, \ie, \texttt{Messi}\footnote{\href{https://www.youtube.com/watch?v=sbPrevs6Pd4}{https://www.youtube.com/watch?v=sbPrevs6Pd4}}, \texttt{Hummingbird}\footnote{\href{https://www.youtube.com/watch?v=RtUQ_pz5wlo}{https://www.youtube.com/watch?v=RtUQ\_pz5wlo}}, \texttt{Swinger}\footnote{\href{https://www.youtube.com/watch?v=cfnbyX9G5Rk}{https://www.youtube.com/watch?v=cfnbyX9G5Rk}}, and \texttt{Football}\footnote{\href{https://www.youtube.com/watch?v=EGAuWZYe2No}{https://www.youtube.com/watch?v=EGAuWZYe2No}}. 
%
%
A sequence of color scene is coded by the corresponding shifted random binary masks at each time step and finally summed up to form a snapshot measurement on the color Bayer RGB sensor (with a ``RGGB'' Bayer color filter array)\footnote{
Note these results are different from the ones reported in~\cite{Yuan20CVPR}. The reason is that the measurements are generally in different ways. In~\cite{Yuan20CVPR}, we up-sampled the raw video by putting each color channel as the mosaic R, G1, G2, and B channels. 
This leads to two identical G channels and the reconstructed and the size of demosaiced image is doubled (both in width and height). For example, for UHD color video \texttt{Football} with original image size of $3840\times1644$, the reconstructed video frames have the size of $7680\times3288$ (demosaiced). This is different from the Bayer pattern model described in Sec.~\ref{Sec:jointcsci}. After some researching on the camera design, in this paper, we follow the Bayer RGGB pattern color video SCI model developed in Sec.~\ref{Sec:SCImodel} to generate the new measurements being used in the experiments. We are convinced that this is more appropriate and closer to real color cameras.
}.
To verify the flexibility of the proposed PnP algorithm, we consider different spatial size and different compression rate $B$.
\begin{list}{\labelitemi}{\leftmargin=8pt \topsep=2pt \parsep=2pt}
	\item \texttt{Messi20} color: A $1920\times1080\times3\times20$ video reconstructed from a snapshot.
	\item \texttt{Hummingbird30} color: A $1920\times1080\times3\times30$ video reconstructed from a snapshot.
	\item \texttt{swinger15} color: A $3840\times2160\times3\times15$ video reconstructed from a snapshot.
	\item \texttt{football40} color: A $3840\times1644\times3\times40$ video reconstructed from a snapshot.
\end{list}

Due to the extremely long running time of other algorithms, we hereby only show the results of GAP-TV, PnP-FFDNet and PnP-FastDVDnet; both FFDNet and FastDVDnet used color denoising as in the mid-scale benchmark data. Note that only grayscale FFDNet denoising was used in \cite{Yuan20CVPR}.

{Figure~\ref{fig:comp_largescale} plots selected reconstruction frames of these three algorithms, where we can see that $i$) due to many fine details, GAP-TV cannot provide high quality results, $ii$) both PnP-FFDNet and PnP-FastDVDnet lead to significant improvements over GAP-TV (at least 3.69 dB in PSNR), and $iii$) PnP-FastDVDnet leads to best results on the first three datasets and for the last one, {\texttt{football40}}, it is 0.39dB lower than PnP-FFDNet. This might due to the crowed players in the scene, which is in favor of FFDNet denoising.}
Importantly, for all these large-scale dataset, with a compression rate varying from 15 to 40, we can all get the reconstruction up to (or at least close to) 30dB. This proves that the video SCI can be used in our daily life videos.

Regarding the running time, as shown in Table~\ref{Table:time_largescale}, 
for all these large-scale datasets, PnP with grayscale denoising (PnP-FFDNet-gray and PnP-FastDVDnet-gray) can finish the reconstruction within one hour. However, when the color denoising algorithms are used, the running time is $10\times$ longer. Again, as mentioned in the simulation, most of the time is consumed by the demosaicing algorithms and we expect a robust deep demosaicing network can speed up the reconstruction. Due to this, the running time of  PnP-FastDVDnet-color is very similar to PnP-FFDNet-color. 
Even this, for the HD ($1920\times1080\times3$) video data with $B$ up to 30, the reconstruction can be finished within 1 hour, but the other algorithms such as DeSCI are not feasible as it will take days. For the UHD ($3840\times1644\times3$) videos, even at $B=40$, the reconstruction can be finished in hours. Note that since spatial pixels in video SCI are decoupled, we can also use multiple CPUs or GPUs performing on blocks rather than the entire frame to speed up the reconstruction. 

\begin{figure}[!htbp!]
	\begin{center}
		\includegraphics[width=0.8\linewidth]{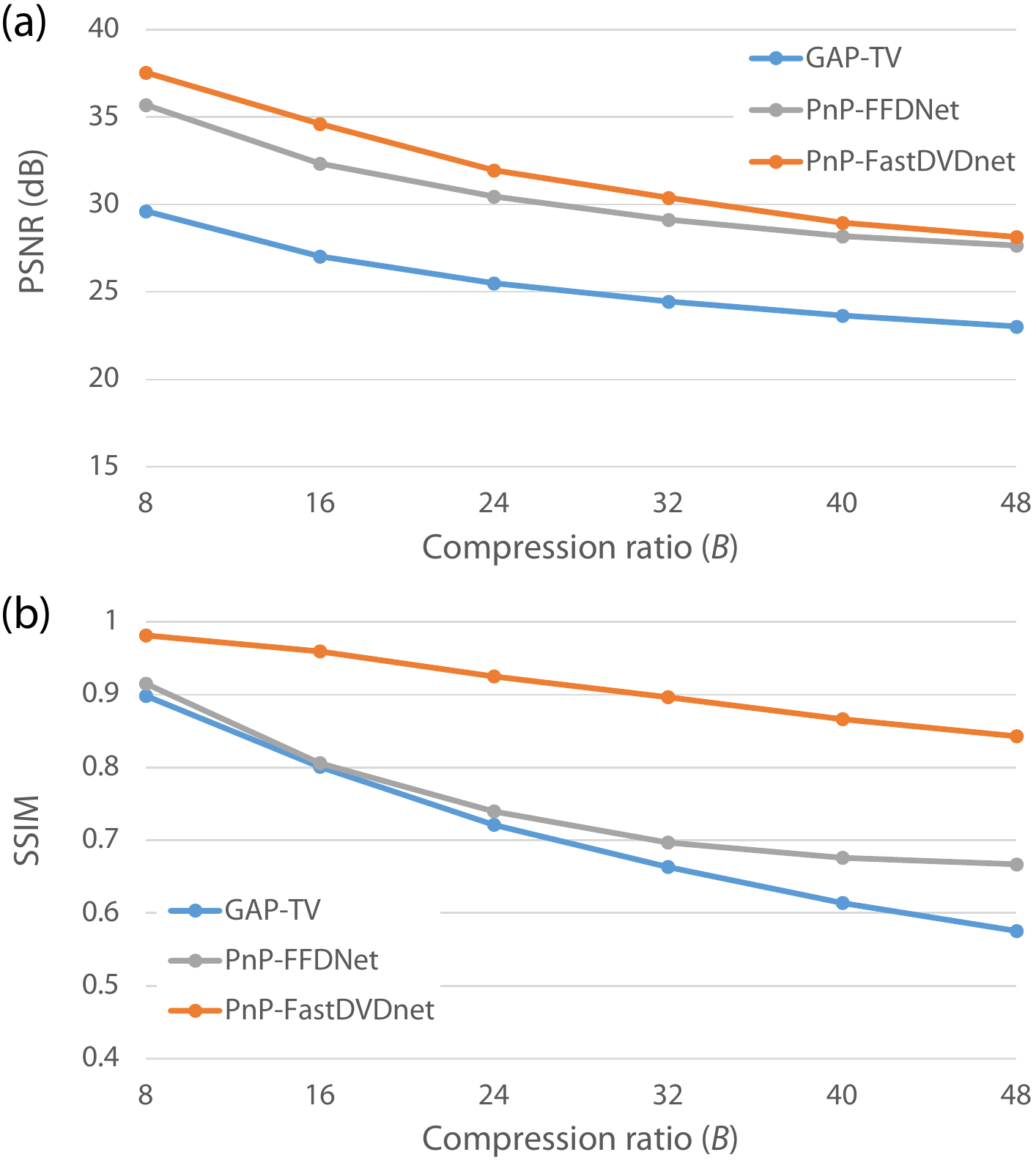}
	\end{center}
	\vspace{-5mm}
	\caption{Reconstruction quality (PSNR in dB (a) and SSIM (b), higher is better) varying compression rates $B$ from 8 to 48 of the proposed PnP methods (PnP-FFDNet and PnP-FastDVDnet) and GAP-TV~\cite{Yuan16ICIP_GAP} .}
	\label{fig:quality_vary_codenum}
\end{figure} %

In addition to these large-scale data with different spatial sizes and compression rates, another way to construct large-scale data for a specific SCI system is to fix the spatial size, but with various compression rates. In this case, the data scales with $B$, which is also challenging to other algorithms including deep learning ones\footnote{In~\cite{Qiao2020_APLP}, it is failed to train a deep neural networks for $B>30$ even with a spatial size $512\times512$ due to the limited GPU memory.}. 
Hereby, we conduct simulation of the \texttt{Hummingbird} data with different compression rates $B = {8,16,24,32, 40,48}$ with results shown in Fig.~\ref{fig:quality_vary_codenum}. It can be seen that even at $B$=48, both PnP-FFDNet and PnP-FastDVDnet can reconstruct the video at PSNR close to 27dB; regarding SSIM, PnP-FastDVDnet achieves 0.85 at $B$=48, which is $>$0.15 higher than PnP-FFDNet and $>$0.25 higher than GAP-TV. Therefore, our proposed PnP algorithms are robust and flexible to different compression rates. This will be further verified by the real data in Sec.~\ref{Sec:gray_hand}.

%

%

\section{Real Data \label{Sec:realdata}}


We now apply the proposed PnP framework to real data captured by SCI cameras to verify the robustness of the proposed algorithms. 
Different data captured by different video SCI cameras are used and these data are of different spatial size, different compression rate and using different modulation patterns. 
We first verify PnP by using grayscale data~\cite{Patrick13OE,Sun17OE} with a fixed $B$; then we conduct the experiments of grayscale data with different compression rates captured by the same system~\cite{Qiao2020_CACTI}. Lastly, we show the results of color data captured by the SCI system in~\cite{Yuan14CVPR}.
Note that, in Sec.~\ref{Sec:gray_hand}, for the first time, we show that a $512\times512\times50$ {\texttt{Hand}} video reconstructed from a snapshot in high quality with each frame having a motion.

\begin{figure}[!htbp!]
	\begin{center}
		\includegraphics[width=1\linewidth]{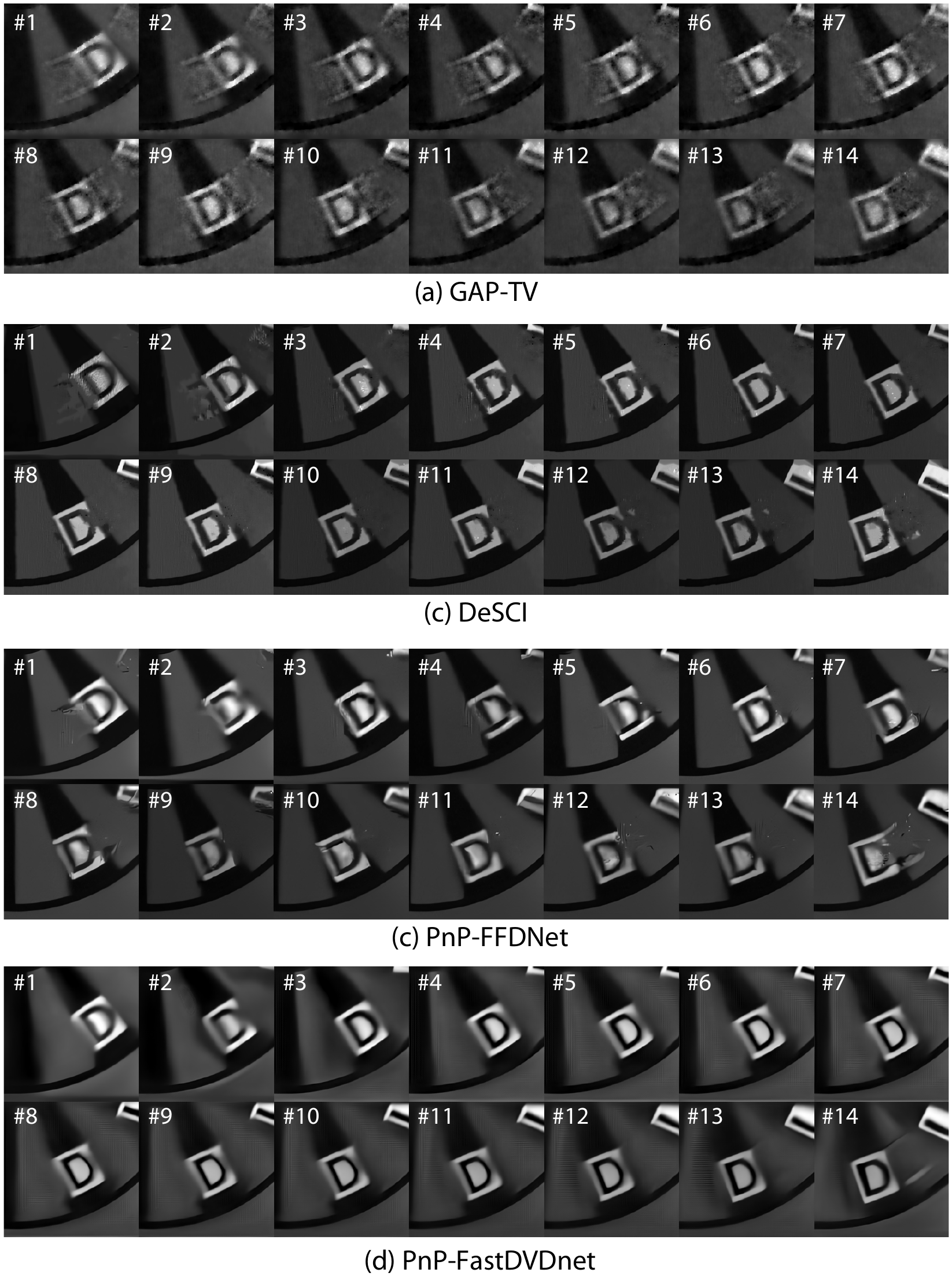}\\
	\end{center}
	\vspace{-5mm}
	\caption{Real data: \texttt{chopper wheel} ($256\times256\times14$). }
	\vspace{-4mm}
	\label{fig:real_chopperwheel}
\end{figure}


\subsection{Grayscale Videos with Fixed Compression Rate}
In this subsection, we verify the proposed PnP algorithm by the following data:
\begin{list}{\labelitemi}{\leftmargin=8pt \topsep=2pt \parsep=2pt}
    \item {\texttt{Chopper wheel}} data captured by the original CACTI paper~\cite{Patrick13OE} is of spatial size $256\times256$ and $B=14$. The results of GAP-TV, DeSCI, PnP-FFDNet and PnP-FastDVDnet are shown in Fig.~\ref{fig:real_chopperwheel}, where we can see that though DeSCI, PnP-FFDNet and PnP-FastDVDnet can all provide good results. Due to the temporal correlation of video investigated in FastDVDnet, the results of PnP-FastDVDnet is having a consistent brightness and of smooth motion.
    \item {\texttt{UCF data}} captured by the video SCI system built in~\cite{Sun17OE} is of large size $1100\times 850$ with $B=10$. The results are shown in Fig.~\ref{fig:real_ucf}, which has a complicated background and a dropping ball on the left.
    It can be seen clearly that PnP-FastDVDnet provides a clean background with fine details. 
\end{list}
Again, since these data are of different sizes and compression rates, it is challenging to use the recently developed end-to-end deep neural networks~\cite{Cheng20ECCV_BIRNAT} to perform all the tasks. For instance, the training time for each task will be of weeks and it consumes a significant amount of power and memory to train the network for large-scale data such as {\texttt{UCF}}.
By contrast, in our proposed PnP framework, the same pre-trained FFDNet or FastDVDnet is used for all these tasks and the results are obtained in seconds. 

\begin{figure}[htbp!]
	\begin{center}
		\includegraphics[width=1.0\linewidth]{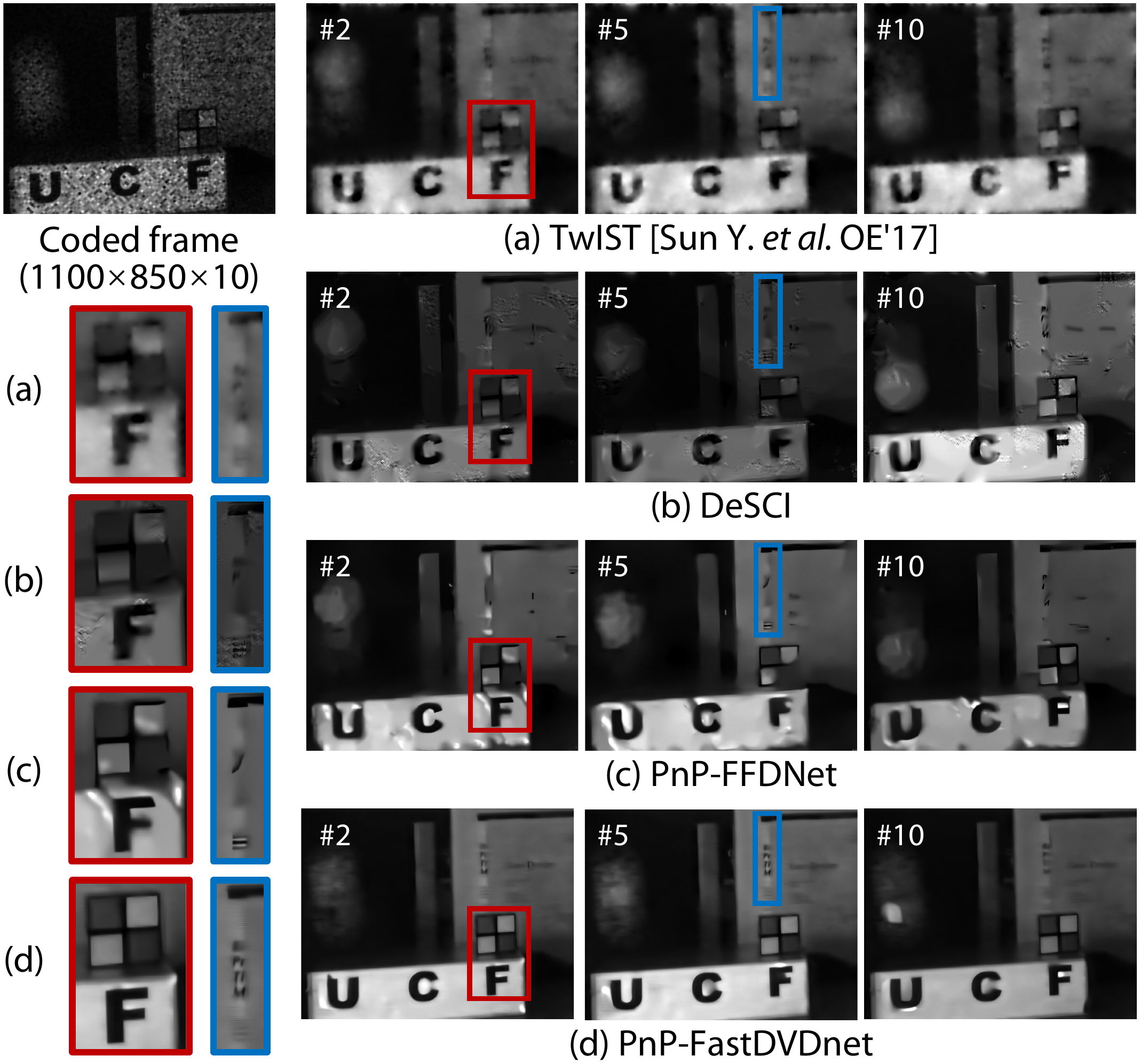}
	\end{center}
	\vspace{-4mm}
	\caption{Real data: \texttt{UCF} high-speed video SCI ($1100\times850\times10$).}
	\label{fig:real_ucf}
	\vspace{-3mm}
\end{figure}

\begin{table}[htbp!]
\caption{Running time (seconds) of real data using different algorithms.}
\vspace{-3mm}
\resizebox{1\columnwidth}{!} {
\begin{tabular}{c cV{3}cccc}
\hlineB{3}
Real dataset  & Pixel resolution & {GAP-TV} & {DeSCI} & {PnP-FFDNet} & {PnP-FastDVDnet} \\ \hlineB{3}
\texttt{chopperwheel} & $256\times256\times14$       & 11.6                        & 3185.8                     & \textbf{2.7}                             & 18.3   

\\ \hline
\texttt{hammer} color & $512\times512\times22$       & 94.5                        & 4791.0                     & \textbf{12.6}                            & 136.6           \\ \hline
\texttt{UCF}          & $1100\times850\times10$      & 300.8                       & 2938.8                    & \textbf{12.5}                            & 132.6                                \\ 
\hlineB{3}
\texttt{hand10}          & $512\times512\times10$      & 37.8                       & 2880.0                    & \textbf{19.3}                            & 29.5            \\ 
\texttt{hand20}          & $512\times512\times20$      & 88.7                       & 4320.0                    & \textbf{42.4}                            & 63.9            \\ 
\texttt{hand30}          & $512\times512\times30$      & 163.0                       & 6120.0                    & \textbf{74.7}                            & 107.7            \\ 
\texttt{hand50}          & $512\times512\times50$      & 303.4                       & 12600.0                    & \textbf{144.5}                            & 203.9            \\ 
\hlineB{3}
\end{tabular}
}
\label{Table:time_real}
\vspace{-2mm}
\end{table}

\begin{figure*}[htbp!]
	\begin{center}
		\includegraphics[width=1.0\linewidth]{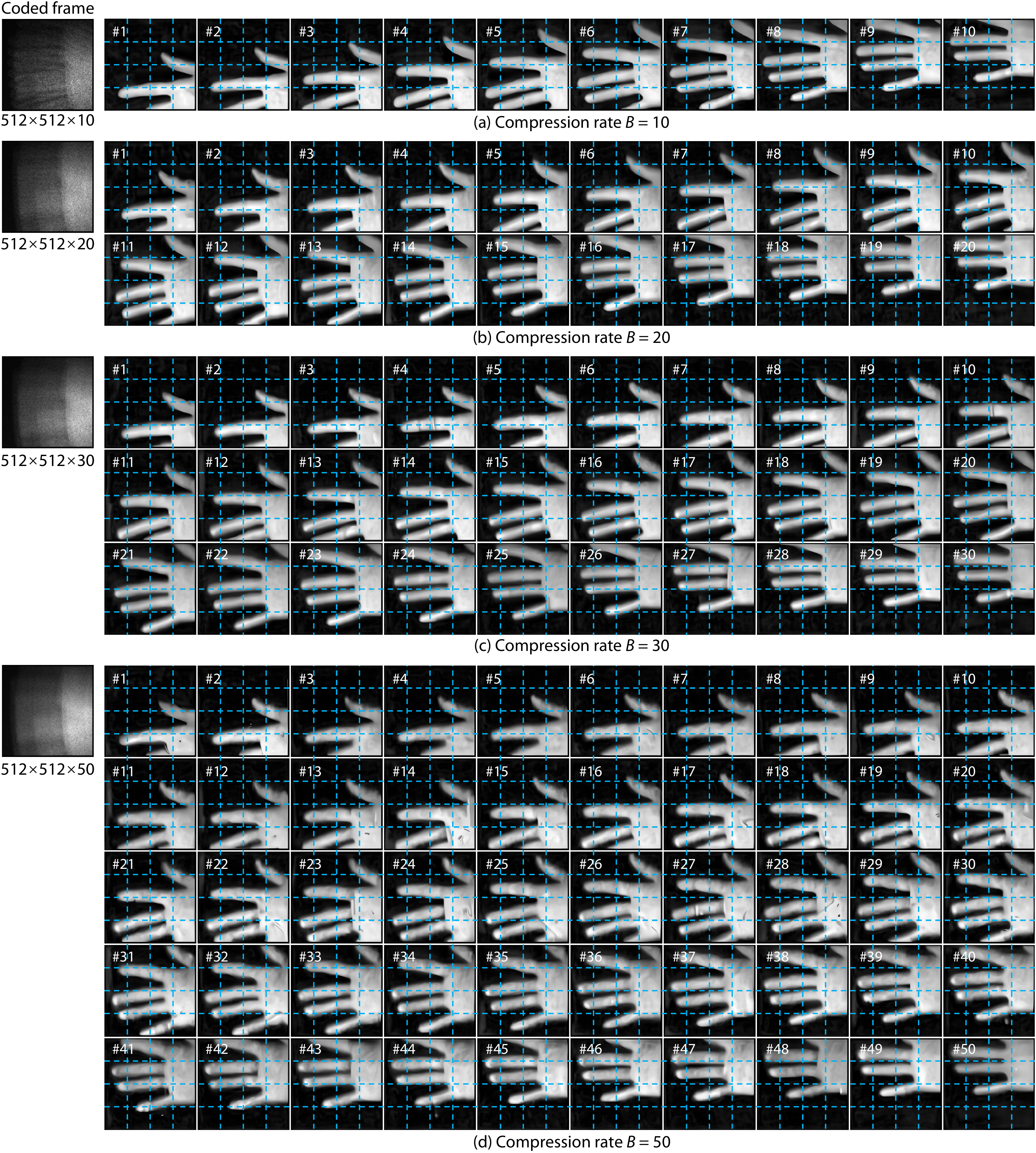}
	\end{center}
	\vspace{-3mm}
	\caption{Real data: \texttt{Hand} high-speed video SCI ($512\times512\times B$) with compression rates, $B$, vary from 10 to 50. Dashed grids are added to aid the visualization of motion details. PnP-FastDVDnet is used for the reconstruction.}
	\label{fig:real_hand}
\end{figure*}

The running time of different algorithms for these real data are shown in Table~\ref{Table:time_real}.
We can see that PnP-FFDNet, which only takes a few seconds for the reconstruction of these grayscale datasets, can provide comparable results as DeSCI, which needs hours even when performed in a frame-wise manner. PnP-FFDNet is significantly better than the speed runner-up GAP-TV (for the top two datasets) in terms of motion-blur reduction and detail preservation, as shown in Figs.~\ref{fig:real_chopperwheel} and \ref{fig:real_ucf}. PnP-FFDNet is at least more than $4\times$ faster than GAP-TV and when the data size is getting larger, the running time increases slower than GAP-TV. In this way, PnP algorithms for SCI achieves a good balance of efficiency and flexibility and PnP-FFDNet could serve as a baseline for SCI recovery.

PnP-FastDVDnet costs about $10\times$ longer than PnP-FFDNet (upper part in Table~\ref{Table:time_real}), which is the price for a higher reconstruction quality.
We also notice that the running time of PnP-FastDVDnet for large-scale data {\texttt{UCF}} is shorter than GAP-TV. This shows another gain of PnP based algorithm, \ie, ready to scale up. This will be further verified by the following {\texttt{Hand}} data.
Therefore, we recommend the PnP-FFDNet as a new baseline, and if a higher quality result is desired, PnP-FastDVDnet is a good choice with a longer running time (but still $20\times$ shorter than DeSCI).

\subsection{Grayscale Videos with Various Compression Rates }
\label{Sec:gray_hand}
Next, we test the PnP algorithms by the data captured by a recently built video SCI system in~\cite{Qiao2020_APLP}, where similar scenes were captured by different compression rates, \ie, $B=\{10,20,30,40,50\}$\footnote{Data at: \url{https://github.com/mq0829/DL-CACTI}.}. 
Unlike the data reported in the previous subsection, here the data are of the same spatial size $512\times 512$, but a compression rate of 50 will stress out the GPU memory in deep neural networks.
{As shown in Fig.~\ref{fig:quality_vary_codenum}. This is another way to construct the large-scale data.}
{Another challenge in video SCI is that though high compression rate results have been reported before, whether the reconstructed video can resolve such a high-speed motion is still a question.}

{To address these concerns, hereby, we show the reconstruction of {\texttt{Hand}} data with $B=10,20,30,50$ in Fig.~\ref{fig:real_hand}, where we can see that  at $B=50$, each frame is different from the previous one, and this leads to a high-speed motion to at least a few pixels motion per reconstructed frame.}
Regarding the running time, it can be seen from Table~\ref{Table:time_real} that both PnP-FFDNet and PnP-FastDVDnet are faster than GAP-TV. At $B=50$, PnP-FFDNet can finish the reconstruction of one measurement in 2.4 minutes and PnP-FastDVDnet needs 3.4 minutes but providing better results. By contrast, DeSCI needs 210 minutes (3.5 hours) to reconstruct 50 frames from a snapshot.

We have also tried to reconstruct this video by training the networks proposed in~\cite{Qiao2020_APLP} and \cite{Cheng20ECCV_BIRNAT}; however, the quality of the results from the trained networks for this {\texttt{Hand}} dataset are poor when $B>20$, mainly due to the high-speed motions.
By contrast, in Fig.~\ref{fig:real_hand}, we can see clear details are reconstructed by the proposed PnP-FastDVDnet.  
Therefore, it is confident to state that the video SCI system along with our proposed PnP algorithm can achieve a compression rate of 50.
When the camera is working at 50 fps~\cite{Qiao2020_CACTI}, the built system can be used to capture high-speed videos at 2500 fps with high quality reconstruction. 
Due to the space limit, we do not show results of other data here and more results can be found in the supplementary videos.

\begin{figure}[!htbp]
	\begin{center}
		\includegraphics[width=1\linewidth]{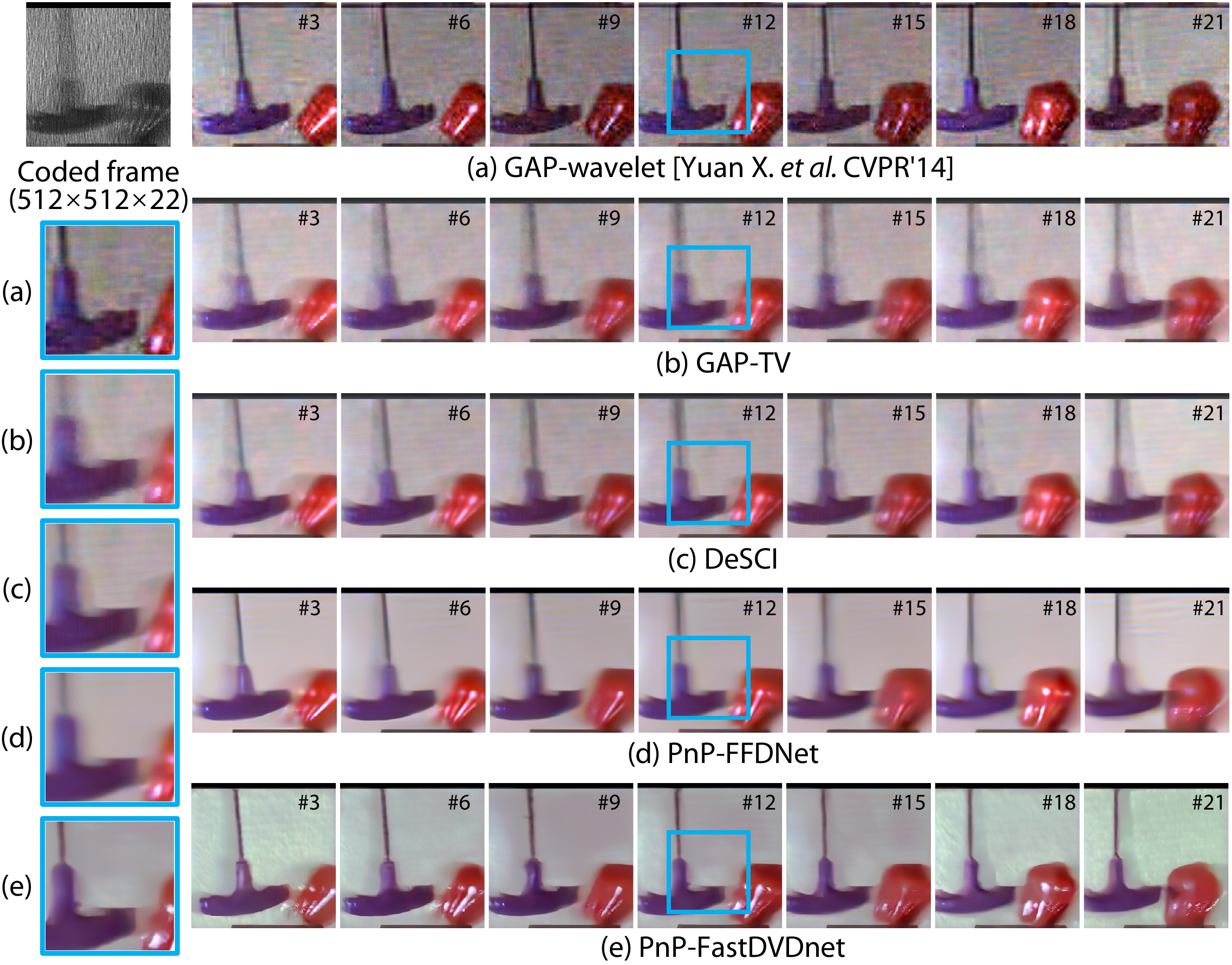}
	\end{center}
	\vspace{-3mm}
	\caption{Real data: \texttt{Hammer} color video SCI ($512\times512\times 3\times22$).}
	\vspace{-4mm}
	\label{fig:real_color_hammer}
\end{figure}

\subsection{Color Videos}
Lastly, we verify the proposed algorithm on the color video SCI data captured by~\cite{Yuan14CVPR}, which has the same model as described in Section~\ref{Sec:jointcsci}. 
Following the procedure in the mid-scale color data, an RGB video of $B=22$ frames with size of $512\times 512 \times 3$ is reconstructed from a single Bayer mosaic measurement shown in Fig.~\ref{fig:real_color_hammer} of the data \texttt{hammer}. 
Along with the running time in Table~\ref{Table:time_real}, we can see that PnP-FFDnet, which only takes about 12 seconds for reconstruction, can provide comparable results as DeSCI, which needs hours. 
GAP-wavelet~\cite{Yuan14CVPR} cannot remove noise in the background and GAP-TV shows blurry results.
PnP-FFDnet shows sharper edges than DeSCI with a clean background.
PnP-FastDVDnet reconstructs sharper boundaries and finer details of the hammer than DeSCI and PnP-FFDNet, but needs 136 seconds (2.27 minutes) for the reconstruction.
We do notice the greenish background of PnP-FastDVDnet, which may come from the smoothing artifacts of brightness across different frames.
%
We will test more color video SCI data using the proposed PnP algorithms in the future. 

\section{Conclusions \label{Sec:Con}}
We proposed plug-and-play algorithms for the reconstruction of snapshot compressive video imaging systems. By integrating deep denoisers into the PnP framework, we not only get excellent results on both simulation and real datasets, but also provide reconstruction in a short time with sufficient flexibility. 
Convergence results of PnP-GAP are proved and we first time show that SCI can be used in large-scale (HD, FHD and UHD) daily life videos. This paves the way of practical applications of SCI.

Regarding the future work, one direction is 
to incorporating an efficient demosaicing network to speed up the reconstruction and also improve the video quality. The other direction is to build a real large-scale video SCI system to be used in advanced cameras~\cite{brady2012multiscale,Brady:20}.

\section*{Acknowledgments.} The work of Jinli Suo and Qionghai Dai is partially supported by NSFC 61722110, 61931012, 61631009 and Beijing Municipal Science \& Technology Commission (BMSTC) (No. Z181100003118014). X. Yuan and Y. Liu contribute
equally to this paper.

\bibliographystyle{IEEEtran}
\bibliography{reference_ECCV.bib,reference_sideinfor.bib,reference_xin.bib}

\begin{thebibliography}{10}
\providecommand{\url}[1]{#1}
\csname url@samestyle\endcsname
\providecommand{\newblock}{\relax}
\providecommand{\bibinfo}[2]{#2}
\providecommand{\BIBentrySTDinterwordspacing}{\spaceskip=0pt\relax}
\providecommand{\BIBentryALTinterwordstretchfactor}{4}
\providecommand{\BIBentryALTinterwordspacing}{\spaceskip=\fontdimen2\font plus
\BIBentryALTinterwordstretchfactor\fontdimen3\font minus
  \fontdimen4\font\relax}
\providecommand{\BIBforeignlanguage}[2]{{%
\expandafter\ifx\csname l@#1\endcsname\relax
\typeout{** WARNING: IEEEtran.bst: No hyphenation pattern has been}%
\typeout{** loaded for the language `#1'. Using the pattern for}%
\typeout{** the default language instead.}%
\else
\language=\csname l@#1\endcsname
\fi
#2}}
\providecommand{\BIBdecl}{\relax}
\BIBdecl

\bibitem{Altmanneaat2298}
Y.~Altmann, S.~McLaughlin, M.~J. Padgett, V.~K. Goyal, A.~O. Hero, and
  D.~Faccio, ``Quantum-inspired computational imaging,'' \emph{Science}, vol.
  361, no. 6403, 2018.

\bibitem{Mait18CI}
J.~N.~Mait, G.~W. Euliss, and R.~A. Athale, ``Computational imaging,''
  \emph{Adv. Opt. Photon.}, vol.~10, no.~2, pp. 409--483, Jun 2018.

\bibitem{BradyNature12}
D.~J. Brady, M.~E. Gehm, R.~A. Stack, D.~L. Marks, D.~S. Kittle, D.~R. Golish,
  E.~M. Vera, and S.~D. Feller, ``Multiscale gigapixel photography,''
  \emph{Nature}, vol. 486, no. 7403, pp. 386--389, 2012.

\bibitem{Brady18Optica}
D.~J. Brady, W.~Pang, H.~Li, Z.~Ma, Y.~Tao, and X.~Cao, ``Parallel cameras,''
  \emph{Optica}, vol.~5, no.~2, 2018.

\bibitem{Ouyang2018DeepLM}
W.~Ouyang, A.~I. Aristov, M.~Lelek, X.~F. Hao, and C.~Zimmer, ``Deep learning
  massively accelerates super-resolution localization microscopy,''
  \emph{Nature Biotechnology}, vol.~36, pp. 460--468, 2018.

\bibitem{Brady15AOP}
D.~J. Brady, A.~Mrozack, K.~MacCabe, and P.~Llull, ``Compressive tomography,''
  \emph{Advances in Optics and Photonics}, vol.~7, no.~4, p. 756, 2015.

\bibitem{Tsai15OL}
T.-H. Tsai, P.~Llull, X.~Yuan, L.~Carin, and D.~J. Brady, ``Spectral-temporal
  compressive imaging,'' \emph{Optics Letters}, vol.~40, no.~17, pp.
  4054--4057, Sep 2015.

\bibitem{Patrick13OE}
P.~Llull, X.~Liao, X.~Yuan, J.~Yang, D.~Kittle, L.~Carin, G.~Sapiro, and D.~J.
  Brady, ``Coded aperture compressive temporal imaging,'' \emph{Optics
  Express}, vol.~21, no.~9, pp. 10\,526--10\,545, 2013.

\bibitem{Wagadarikar08CASSI}
A.~Wagadarikar, R.~John, R.~Willett, and D.~Brady, ``Single disperser design
  for coded aperture snapshot spectral imaging,'' \emph{Applied Optics},
  vol.~47, no.~10, pp. B44--B51, 2008.

\bibitem{Yuan14CVPR}
X.~Yuan, P.~Llull, X.~Liao, J.~Yang, D.~J. Brady, G.~Sapiro, and L.~Carin,
  ``Low-cost compressive sensing for color video and depth,'' in \emph{IEEE
  Conference on Computer Vision and Pattern Recognition (CVPR)}, 2014, Journal
  Article, pp. 3318--3325.

\bibitem{Wagadarikar09CASSI}
A.~A. Wagadarikar, N.~P. Pitsianis, X.~Sun, and D.~J. Brady, ``Video rate
  spectral imaging using a coded aperture snapshot spectral imager,''
  \emph{Optics Express}, vol.~17, no.~8, pp. 6368--6388, 2009.

\bibitem{Hitomi11ICCV}
Y.~Hitomi, J.~Gu, M.~Gupta, T.~Mitsunaga, and S.~K. Nayar, ``Video from a
  single coded exposure photograph using a learned over-complete dictionary,''
  in \emph{2011 International Conference on Computer Vision}.\hskip 1em plus
  0.5em minus 0.4em\relax IEEE, 2011, pp. 287--294.

\bibitem{Reddy11CVPR}
D.~Reddy, A.~Veeraraghavan, and R.~Chellappa, ``{\rm P2C2}: Programmable pixel
  compressive camera for high speed imaging,'' in \emph{IEEE Conference on
  Computer Vision and Pattern Recognition (CVPR)}, Conference Proceedings, pp.
  329--336.

\bibitem{Yuan16BOE}
X.~Yuan and S.~Pang, ``Structured illumination temporal compressive
  microscopy,'' \emph{Biomedical Optics Express}, vol.~7, pp. 746--758, 2016.

\bibitem{Deng19_sin}
C.~{Deng}, Y.~{Zhang}, Y.~{Mao}, J.~{Fan}, J.~{Suo}, Z.~{Zhang}, and Q.~{Dai},
  ``Sinusoidal sampling enhanced compressive camera for high speed imaging,''
  \emph{IEEE Transactions on Pattern Analysis and Machine Intelligence}, pp.
  1--1, 2019.

\bibitem{Gehm07}
M.~E. Gehm, R.~John, D.~J. Brady, R.~M. Willett, and T.~J. Schulz,
  ``Single-shot compressive spectral imaging with a dual-disperser
  architecture,'' \emph{Optics Express}, vol.~15, no.~21, pp. 14\,013--14\,027,
  2007.

\bibitem{Miao19ICCV}
X.~Miao, X.~Yuan, Y.~Pu, and V.~Athitsos, ``$\lambda$-net: Reconstruct
  hyperspectral images from a snapshot measurement,'' in \emph{IEEE/CVF
  Conference on Computer Vision (ICCV)}, 2019.

\bibitem{Yuan15JSTSP}
X.~Yuan, T.-H. Tsai, R.~Zhu, P.~Llull, D.~Brady, and L.~Carin, ``Compressive
  hyperspectral imaging with side information,'' \emph{IEEE Journal of Selected
  Topics in Signal Processing}, vol.~9, no.~6, pp. 964--976, September 2015.

\bibitem{Qiao2020_APLP}
M.~Qiao, Z.~Meng, J.~Ma, and X.~Yuan, \emph{APL Photonics}, vol.~5, no.~3, p.
  030801, 2020.

\bibitem{Sun17OE}
Y.~Sun, X.~Yuan, and S.~Pang, ``Compressive high-speed stereo imaging,''
  \emph{Opt Express}, vol.~25, no.~15, pp. 18\,182--18\,190, 2017.

\bibitem{Bioucas-Dias2007TwIST}
J.~M. Bioucas-Dias and M.~A. Figueiredo, ``A new twist: Two-step iterative
  shrinkage/thresholding algorithms for image restoration,'' \emph{IEEE
  Transactions on Image processing}, vol.~16, no.~12, pp. 2992--3004, 2007.

\bibitem{Yang14GMMonline}
J.~Yang, X.~Liao, X.~Yuan, P.~Llull, D.~J. Brady, G.~Sapiro, and L.~Carin,
  ``Compressive sensing by learning a {G}aussian mixture model from
  measurements,'' \emph{IEEE Transaction on Image Processing}, vol.~24, no.~1,
  pp. 106--119, January 2015.

\bibitem{Yang14GMM}
J.~Yang, X.~Yuan, X.~Liao, P.~Llull, G.~Sapiro, D.~J. Brady, and L.~Carin,
  ``Video compressive sensing using {G}aussian mixture models,'' \emph{IEEE
  Transaction on Image Processing}, vol.~23, no.~11, pp. 4863--4878, November
  2014.

\bibitem{Yuan16ICIP_GAP}
X.~Yuan, ``Generalized alternating projection based total variation
  minimization for compressive sensing,'' in \emph{2016 IEEE International
  Conference on Image Processing (ICIP)}, Sept 2016, pp. 2539--2543.

\bibitem{Liao14GAP}
X.~Liao, H.~Li, and L.~Carin, ``Generalized alternating projection for
  weighted-$\ell_{2,1}$ minimization with applications to model-based
  compressive sensing,'' \emph{SIAM Journal on Imaging Sciences}, vol.~7,
  no.~2, pp. 797--823, 2014.

\bibitem{Liu18TPAMI}
Y.~Liu, X.~Yuan, J.~Suo, D.~Brady, and Q.~Dai, ``Rank minimization for snapshot
  compressive imaging,'' \emph{IEEE Transactions on Pattern Analysis and
  Machine Intelligence}, vol.~41, no.~12, pp. 2990--3006, Dec 2019.

\bibitem{Ma19ICCV}
J.~Ma, X.~Liu, Z.~Shou, and X.~Yuan, ``Deep tensor admm-net for snapshot
  compressive imaging,'' in \emph{IEEE/CVF Conference on Computer Vision
  (ICCV)}, 2019.

\bibitem{Li2020ICCP}
Y.~{Li}, M.~{Qi}, R.~{Gulve}, M.~{Wei}, R.~{Genov}, K.~N. {Kutulakos}, and
  W.~{Heidrich}, ``End-to-end video compressive sensing using
  anderson-accelerated unrolled networks,'' in \emph{2020 IEEE International
  Conference on Computational Photography (ICCP)}, 2020, pp. 1--12.

\bibitem{Cheng20ECCV_BIRNAT}
Z.~Cheng, R.~Lu, Z.~Wang, H.~Zhang, B.~Chen, Z.~Meng, and X.~Yuan, ``Birnat:
  Bidirectional recurrent neural networks with adversarial training for video
  snapshot compressive imaging,'' in \emph{European Conference on Computer
  Vision (ECCV)}, August 2020.

\bibitem{Yuan13ICIP}
X.~Yuan, J.~Yang, P.~Llull, X.~Liao, G.~Sapiro, D.~J. Brady, and L.~Carin,
  ``Adaptive temporal compressive sensing for video,'' \emph{IEEE International
  Conference on Image Processing}, pp. 1--4, 2013.

\bibitem{Cao16SPM}
X.~Cao, T.~Yue, X.~Lin, S.~Lin, X.~Yuan, Q.~Dai, L.~Carin, and D.~J. Brady,
  ``Computational snapshot multispectral cameras: Toward dynamic capture of the
  spectral world,'' \emph{IEEE Signal Processing Magazine}, vol.~33, no.~5, pp.
  95--108, 2016.

\bibitem{Wang18PAMI}
L.~{Wang}, Z.~{Xiong}, H.~{Huang}, G.~{Shi}, F.~{Wu}, and W.~{Zeng},
  ``High-speed hyperspectral video acquisition by combining nyquist and
  compressive sampling,'' \emph{IEEE Transactions on Pattern Analysis and
  Machine Intelligence}, vol.~41, no.~4, pp. 857--870, April 2019.

\bibitem{Meng2020_OL_SHEM}
Z.~Meng, M.~Qiao, J.~Ma, Z.~Yu, K.~Xu, and X.~Yuan, ``Snapshot multispectral
  endomicroscopy,'' \emph{Opt. Lett.}, vol.~45, no.~14, pp. 3897--3900, Jul
  2020.

\bibitem{Meng20ECCV_TSAnet}
Z.~Meng, J.~Ma, and X.~Yuan, ``End-to-end low cost compressive spectral imaging
  with spatial-spectral self-attention,'' in \emph{European Conference on
  Computer Vision (ECCV)}, August 2020.

\bibitem{Qiao2020_CACTI}
M.~Qiao, X.~Liu, and X.~Yuan, ``Snapshot spatial-temporal compressive
  imaging,'' \emph{Opt. Lett.}, 2020.

\bibitem{Llull15Optica}
P.~Llull, X.~Yuan, L.~Carin, and D.~J. Brady, ``Image translation for
  single-shot focal tomography,'' \emph{Optica}, vol.~2, no.~9, pp. 822--825,
  2015.

\bibitem{Yuan16AO}
X.~Yuan, X.~Liao, P.~Llull, D.~Brady, and L.~Carin, ``Efficient patch-based
  approach for compressive depth imaging,'' \emph{Applied Optics}, vol.~55,
  no.~27, pp. 7556--7564, Sep 2016.

\bibitem{Tsai15OE}
T.-H. Tsai, X.~Yuan, and D.~J. Brady, ``Spatial light modulator based color
  polarization imaging,'' \emph{Optics Express}, vol.~23, no.~9, pp.
  11\,912--11\,926, May 2015.

\bibitem{Candes06ITT}
C.~Emmanuel, J.~Romberg, and T.~Tao, ``Robust uncertainty principles: Exact
  signal reconstruction from highly incomplete frequency information,''
  \emph{IEEE Transactions on Information Theory}, vol.~52, no.~2, pp. 489--509,
  February 2006.

\bibitem{Donoho06ITT}
D.~L. Donoho, ``Compressed sensing,'' \emph{IEEE Transactions on Information
  Theory}, vol.~52, no.~4, pp. 1289--1306, April 2006.

\bibitem{Rudin92_TV}
L.~I. Rudin, S.~Osher, and E.~Fatemi, ``Nonlinear total variation based noise
  removal algorithms,'' \emph{Physica D: Nonlinear Phenomena}, vol.~60, no.
  1-4, pp. 259--268, 1992.

\bibitem{Gu14CVPR}
S.~Gu, L.~Zhang, W.~Zuo, and X.~Feng, ``Weighted nuclear norm minimization with
  application to image denoising,'' in \emph{IEEE Conference on Computer Vision
  and Pattern Recognition (CVPR)}, 2014, pp. 2862--2869.

\bibitem{Boyd11ADMM}
S.~Boyd, N.~Parikh, E.~Chu, B.~Peleato, and J.~Eckstein, ``Distributed
  optimization and statistical learning via the alternating direction method of
  multipliers,'' \emph{Foundations and Trends in Machine Learning}, vol.~3,
  no.~1, pp. 1--122, January 2011.

\bibitem{zhang2017beyond}
K.~Zhang, W.~Zuo, Y.~Chen, D.~Meng, and L.~Zhang, ``Beyond a {Gaussian}
  denoiser: Residual learning of deep {CNN} for image denoising,'' \emph{IEEE
  Transactions on Image Processing}, vol.~26, no.~7, pp. 3142--3155, 2017.

\bibitem{Iliadis18DSPvideoCS}
M.~Iliadis, L.~Spinoulas, and A.~K. Katsaggelos, ``Deep fully-connected
  networks for video compressive sensing,'' \emph{Digital Signal Processing},
  vol.~72, pp. 9--18, 2018.

\bibitem{Jin17TIP}
K.~H. Jin, M.~T. McCann, E.~Froustey, and M.~Unser, ``Deep convolutional neural
  network for inverse problems in imaging,'' \emph{IEEE Transactions on Image
  Processing}, vol.~26, no.~9, pp. 4509--4522, Sept 2017.

\bibitem{Kulkarni2016CVPR}
K.~Kulkarni, S.~Lohit, P.~Turaga, R.~Kerviche, and A.~Ashok, ``Reconnet:
  Non-iterative reconstruction of images from compressively sensed random
  measurements,'' in \emph{CVPR}, 2016.

\bibitem{LearningInvert2017}
A.~Mousavi and R.~G. Baraniuk, ``Learning to invert: Signal recovery via deep
  convolutional networks,'' in \emph{2017 IEEE International Conference on
  Acoustics, Speech and Signal Processing (ICASSP)}, March 2017, pp.
  2272--2276.

\bibitem{George17lensless}
\BIBentryALTinterwordspacing
A.~Sinha, J.~Lee, S.~Li, and G.~Barbastathis, ``Lensless computational imaging
  through deep learning,'' \emph{Optica}, vol.~4, no.~9, pp. 1117--1125, Sep
  2017. [Online]. Available:
  \url{http://www.osapublishing.org/optica/abstract.cfm?URI=optica-4-9-1117}
\BIBentrySTDinterwordspacing

\bibitem{Yuan18OE}
X.~Yuan and Y.~Pu, ``Parallel lensless compressive imaging via deep
  convolutional neural networks,'' \emph{Optics Express}, vol.~26, no.~2, pp.
  1962--1977, Jan 2018.

\bibitem{Yoshida18ECCV}
M.~Yoshida, A.~Torii, M.~Okutomi, K.~Endo, Y.~Sugiyama, R.-i. Taniguchi, and
  H.~Nagahara, ``Joint optimization for compressive video sensing and
  reconstruction under hardware constraints,'' in \emph{The European Conference
  on Computer Vision (ECCV)}, September 2018.

\bibitem{Venkatakrishnan_13PnP}
S.~V. {Venkatakrishnan}, C.~A. {Bouman}, and B.~{Wohlberg}, ``Plug-and-play
  priors for model based reconstruction,'' in \emph{2013 IEEE Global Conference
  on Signal and Information Processing}, 2013, pp. 945--948.

\bibitem{Sreehari16PnP}
S.~{Sreehari}, S.~V. {Venkatakrishnan}, B.~{Wohlberg}, G.~T. {Buzzard}, L.~F.
  {Drummy}, J.~P. {Simmons}, and C.~A. {Bouman}, ``Plug-and-play priors for
  bright field electron tomography and sparse interpolation,'' \emph{IEEE
  Transactions on Computational Imaging}, vol.~2, no.~4, pp. 408--423, 2016.

\bibitem{Chan2017PlugandPlayAF}
S.~H. Chan, X.~Wang, and O.~A. Elgendy, ``Plug-and-play {ADMM} for image
  restoration: Fixed-point convergence and applications,'' \emph{IEEE
  Transactions on Computational Imaging}, vol.~3, pp. 84--98, 2017.

\bibitem{Ryu2019PlugandPlayMP}
E.~K. Ryu, J.~Liu, S.~Wang, X.~Chen, Z.~Wang, and W.~Yin, ``Plug-and-play
  methods provably converge with properly trained denoisers,'' in \emph{ICML},
  2019.

\bibitem{Zhang17SPM_deepdenoise}
L.~Zhang and W.~Zuo, ``Image restoration: From sparse and low-rank priors to
  deep priors [lecture notes],'' \emph{IEEE Signal Processing Magazine},
  vol.~34, no.~5, pp. 172--179, 2017.

\bibitem{Zhang18TIP_FFDNet}
\BIBentryALTinterwordspacing
K.~Zhang, W.~Zuo, and L.~Zhang, ``{FFDNet}: Toward a fast and flexible solution
  for {CNN}-based image denoising.'' \emph{IEEE Trans. Image Processing},
  vol.~27, no.~9, pp. 4608--4622, 2018. [Online]. Available:
  \url{http://dblp.uni-trier.de/db/journals/tip/tip27.html#ZhangZZ18}
\BIBentrySTDinterwordspacing

\bibitem{Yuan20CVPR}
X.~Yuan, Y.~Liu, J.~Suo, and Q.~Dai, ``Plug-and-play algorithms for large-scale
  snapshot compressive imaging,'' in \emph{IEEE Conference on Computer Vision
  and Pattern Recognition (CVPR)}, 2020.

\bibitem{Jalali19TIT_SCI}
S.~Jalali and X.~Yuan, ``Snapshot compressed sensing: Performance bounds and
  algorithms,'' \emph{IEEE Transactions on Information Theory}, vol.~65,
  no.~12, pp. 8005--8024, Dec 2019.

\bibitem{Tassano_2020_CVPR}
M.~Tassano, J.~Delon, and T.~Veit, ``Fastdvdnet: Towards real-time deep video
  denoising without flow estimation,'' in \emph{Proceedings of the IEEE/CVF
  Conference on Computer Vision and Pattern Recognition (CVPR)}, June 2020.

\bibitem{Candes05compressed}
E.~Cand\'es, J.~Romberg, and T.~Tao, ``Signal recovery from incomplete and
  inaccurate measurements,'' \emph{Comm. Pure Appl. Math}, vol.~59, no.~8, pp.
  1207--1223, 2005.

\bibitem{donoho2006compressed}
D.~L. Donoho \emph{et~al.}, ``Compressed sensing,'' \emph{IEEE Transactions on
  information theory}, vol.~52, no.~4, pp. 1289--1306, 2006.

\bibitem{Jalali18ISIT}
S.~Jalali and X.~Yuan, ``Compressive imaging via one-shot measurements,'' in
  \emph{IEEE International Symposium on Information Theory (ISIT)}, 2018.

\bibitem{Tassano_19ICIP_DVDnet}
M.~{Tassano}, J.~{Delon}, and T.~{Veit}, ``Dvdnet: A fast network for deep
  video denoising,'' in \emph{2019 IEEE International Conference on Image
  Processing (ICIP)}, 2019, pp. 1805--1809.

\bibitem{Beck09IST}
A.~Beck and M.~Teboulle, ``A fast iterative shrinkage-thresholding algorithm
  for linear inverse problems,'' \emph{SIAM J. Img. Sci.}, vol.~2, no.~1, pp.
  183--202, Mar. 2009.

\bibitem{Zheng20_PRJ_PnP-CASSI}
S.~Zheng, Y.~Liu, Z.~Meng, M.~Qiao, Z.~Tong, X.~Yang, S.~Han, and X.~Yuan,
  ``Deep plug-and-play priors for spectral snapshot compressive imaging,''
  \emph{Photonics Research}, vol.~9, Jan 2021.

\bibitem{Zha2020TIP_JPG}
Z.~{Zha}, X.~{Yuan}, B.~{Wen}, J.~{Zhang}, J.~{Zhou}, and C.~{Zhu}, ``Image
  restoration using joint patch-group-based sparse representation,'' \emph{IEEE
  Transactions on Image Processing}, vol.~29, pp. 7735--7750, 2020.

\bibitem{Zha2020TIP_NSSP}
Z.~{Zha}, X.~{Yuan}, J.~{Zhou}, C.~{Zhu}, and B.~{Wen}, ``Image restoration via
  simultaneous nonlocal self-similarity priors,'' \emph{IEEE Transactions on
  Image Processing}, vol.~29, pp. 8561--8576, 2020.

\bibitem{Stanley05_TV}
S.~Osher, M.~Burger, D.~Goldfarb, J.~Xu, and W.~Yin, ``An iterative
  regularization method for total variation-based image restoration,''
  \emph{Multiscale Modeling \& Simulation}, vol.~4, no.~2, pp. 460--489, 2005.

\bibitem{yang2013efficient}
S.~Yang, J.~Wang, W.~Fan, X.~Zhang, P.~Wonka, and J.~Ye, ``An efficient admm
  algorithm for multidimensional anisotropic total variation regularization
  problems,'' in \emph{Proceedings of the 19th ACM SIGKDD international
  conference on Knowledge discovery and data mining}, 2013, pp. 641--649.

\bibitem{Crouse98_wavelet}
M.~S. {Crouse}, R.~D. {Nowak}, and R.~G. {Baraniuk}, ``Wavelet-based
  statistical signal processing using hidden markov models,'' \emph{IEEE
  Transactions on Signal Processing}, vol.~46, no.~4, pp. 886--902, 1998.

\bibitem{Starck02_Curvelet}
{Jean-Luc Starck}, E.~J. {Candes}, and D.~L. {Donoho}, ``The curvelet transform
  for image denoising,'' \emph{IEEE Transactions on Image Processing}, vol.~11,
  no.~6, pp. 670--684, 2002.

\bibitem{Aharon06TSP}
M.~Aharon, M.~Elad, and A.~Bruckstein, ``{K-SVD}: An algorithm for designing
  overcomplete dictionaries for sparse representation,'' \emph{IEEE
  Transactions on Signal Processing}, vol.~54, no.~11, pp. 4311--4322, 2006.

\bibitem{Mairal_09ICCV_LSSC}
J.~{Mairal}, F.~{Bach}, J.~{Ponce}, G.~{Sapiro}, and A.~{Zisserman},
  ``Non-local sparse models for image restoration,'' in \emph{2009 IEEE 12th
  International Conference on Computer Vision}, 2009, pp. 2272--2279.

\bibitem{Gu17IJCV}
S.~Gu, Q.~Xie, D.~Meng, W.~Zuo, X.~Feng, and L.~Zhang, ``Weighted nuclear norm
  minimization and its applications to low level vision,'' \emph{International
  Journal of Computer Vision}, vol. 121, no.~2, pp. 183--208, 2017.

\bibitem{Dabov07BM3D}
K.~Dabov, A.~Foi, V.~Katkovnik, and K.~Egiazarian, ``Image denoising by sparse
  3d transform-domain collaborative filtering,'' \emph{IEEE Transactions on
  Image Processing}, vol.~16, no.~8, pp. 2080--2095, August 2007.

\bibitem{Maggioni2012VideoDD}
M.~Maggioni, G.~Boracchi, A.~Foi, and K.~O. Egiazarian, ``Video denoising,
  deblocking, and enhancement through separable 4-d nonlocal spatiotemporal
  transforms,'' \emph{IEEE Transactions on Image Processing}, vol.~21, pp.
  3952--3966, 2012.

\bibitem{XieNIPS2012_deepDN}
J.~Xie, L.~Xu, and E.~Chen, ``Image denoising and inpainting with deep neural
  networks,'' in \emph{Advances in Neural Information Processing Systems 25},
  2012, pp. 341--349.

\bibitem{zha2020power}
Z.~{Zha}, X.~{Yuan}, J.~T. {Zhou}, J.~{Zhou}, B.~{Wen}, and C.~{Zhu}, ``The
  power of triply complementary priors for image compressive sensing,'' in
  \emph{2020 IEEE International Conference on Image Processing (ICIP)}, 2020,
  pp. 983--987.

\bibitem{GharbiACM16}
M.~Gharbi, G.~Chaurasia, S.~Paris, and F.~Durand, ``Deep joint demosaicking and
  denoising,'' \emph{ACM Trans. Graph.}, vol.~35, no.~6, Nov. 2016.

\bibitem{LiDemosaicing08}
X.~Li, B.~Gunturk, and L.~Zhang, ``{Image demosaicing: a systematic survey},''
  in \emph{Visual Communications and Image Processing 2008}, vol. 6822.\hskip
  1em plus 0.5em minus 0.4em\relax SPIE, 2008, pp. 489 -- 503.

\bibitem{Brady:20}
D.~J. Brady, L.~Fang, and Z.~Ma, ``Deep learning for camera data acquisition,
  control, and image estimation,'' \emph{Adv. Opt. Photon.}, vol.~12, no.~4,
  pp. 787--846, Dec 2020.

\bibitem{malvar2004high-quality}
R.~Malvar, L.-w. He, and R.~Cutler, ``High-quality linear interpolation for
  demosaicing of bayer-patterned color images,'' in \emph{International
  Conference of Acoustic, Speech and Signal Processing}, May 2004.

\bibitem{Menon07}
D.~{Menon}, S.~{Andriani}, and G.~{Calvagno}, ``Demosaicing with directional
  filtering and a posteriori decision,'' \emph{IEEE Transactions on Image
  Processing}, vol.~16, no.~1, pp. 132--141, 2007.

\bibitem{Wang04imagequality}
Z.~Wang, A.~C. Bovik, H.~R. Sheikh, E.~P. Simoncelli \emph{et~al.}, ``Image
  quality assessment: From error visibility to structural similarity,''
  \emph{IEEE Transactions on Image Processing}, vol.~13, no.~4, pp. 600--612,
  2004.

\bibitem{brady2012multiscale}
D.~Brady, M.~Gehm, R.~Stack, D.~Marks, D.~Kittle, D.~Golish, E.~Vera, and
  S.~Feller, ``Multiscale gigapixel photography,'' \emph{Nature}, vol. 486, no.
  7403, pp. 386--389, 2012.

\end{thebibliography}

\end{document}